\documentclass[journal]{IEEEtran}
\usepackage[caption=false,font=normalsize,labelfont=sf,textfont=sf]{subfig}
\usepackage{cite}
\usepackage{booktabs}
\usepackage{amsmath,amssymb,amsfonts}
\usepackage{amsthm}
\usepackage{algorithmic}
\usepackage{graphicx}
\usepackage{textcomp}
\usepackage{xcolor}
\usepackage{makecell}
\usepackage{multirow}
\usepackage{todonotes}
\usepackage{tcolorbox}
\usepackage{mathtools}
\usepackage{diagbox}
\usepackage{array}
\usepackage{verbatim}
\usepackage{bm}
\usepackage{upgreek}
\usepackage{tabularx}
\usepackage{cleveref}
\newtheorem{proposition}{Proposition}
\newtheorem{definition}{Definition}
\newtheorem{lemma}{Lemma}

\theoremstyle{definition}
\newtheorem{remark}{Remark}

\allowdisplaybreaks

\makeatletter
\renewcommand{\citepunct}{,\penalty\@m\hskip.13emplus.1emminus.1em}
\renewcommand{\citedash}{\hbox{--}\penalty\@m}
\makeatother

\def\BibTeX{{\rm B\kern-.05em{\sc i\kern-.025em b}\kern-.08em
    T\kern-.1667em\lower.7ex\hbox{E}\kern-.125emX}}
\usepackage{balance}

\begin{document}
\title{Implicit Neural Representation for Multiuser Continuous Aperture Array Beamforming}
\author{Shiyong~Chen,~\IEEEmembership{Student Member,~IEEE,} Shengqian~Han,~\IEEEmembership{Senior Member,~IEEE}, and Jia~Guo,~\IEEEmembership{Member,~IEEE}

\thanks{Shiyong Chen is with the School of Electronics and Information Engineering, Beihang University, Beijing 100191, China (email: shiyongchen@buaa.edu.cn).}
\thanks{Shengqian Han is with the School of Electronics and Information Engineering, Beihang University, Beijing 100191, China (email: sqhan@buaa.edu.cn).}
\thanks{Jia Guo is with the School of Electronic Engineering and Computer
Science, Queen Mary University of London, London E1 4NS, U.K. (email: jia.guo@qmul.ac.uk).}
%\thanks{Manuscript received April 19, 2025; revised August 16, 2025.}
}

%\markboth{Journal of \LaTeX\ Class Files,~Vol.~18, No.~9, September~2020}%
%{How to Use the IEEEtran \LaTeX \ Templates}

\maketitle

\begin{abstract}
This paper studies the optimization of beamforming functions for multiuser multi-continuous aperture array (CAPA) systems, where both the base station and the users are equipped with CAPAs. We first derive a closed-form expression for the achievable sum rate, and then develop a functional weighted minimum mean-squared error (WMMSE) algorithm, which transforms the functional optimization problem into an equivalent parameter optimization problem by employing orthonormal basis expansion. Based on the functional WMMSE algorithm, we further propose BeamINR, an implicit neural representation (INR) method for learning continuous beamforming functions. BeamINR is designed as a graph neural network to exploit the permutation equivariance of the optimal beamforming policy, with an update equation designed according to the functional WMMSE iterations. Simulation results show that both the functional WMMSE algorithm and BeamINR outperform existing numerical and INR-based baselines. BeamINR approaches the sum rate of the functional WMMSE with substantially lower inference latency. Compared with INR-based baselines, BeamINR reduces training complexity and improves generalization to the number of users, CAPA sizes, and carrier~frequencies.
\end{abstract}


\section{Introduction}
\IEEEPARstart{M}{assive} multi-input multi-output (MIMO) is a key technology for improving spectral and energy efficiency~\cite{Massive_MIMO_for}. However, its scalability is constrained by the use of spatially discrete arrays (SPDAs). In fully digital architectures, each antenna requires a dedicated radio-frequency chain, which increases hardware cost and power consumption. Moreover, the conventional half-wavelength element spacing leads to physically large arrays, posing deployment challenges~\cite{Massive_MIMO_ten_myths}. 
Holographic MIMO (HMIMO)~\cite{Holographic_MIMO_Surfaces, Holographic_MIMO_How_Many, Learning_based_Multiuser_Beamforming} and reconfigurable intelligent surfaces~\cite{Reconfigurable_Intelligent_Surfaces, Wireless_Communications_Through_Reconfigurable} 
alleviate these limitations by integrating densely spaced radiating elements onto compact surfaces. 
A continuous aperture array (CAPA) is the limiting case of such an architecture with an infinite number of radiating elements. Unlike conventional SPDA beamforming, CAPA beamforming is defined by a continuous current distribution over the aperture. The resulting beamforming optimization is therefore a functional optimization problem that cannot be solved directly by conventional array-processing methods~\cite{CAPA_based}.

Numerical methods for CAPA beamforming follow two main approaches. Approximation-based methods expand the channel and beamforming functions over a finite Fourier basis and then optimize the resulting coefficients~\cite{Wavenumber, On_the_SE, Pattern_Division}. Although this converts the functional problem into a finite-dimensional form, its approximation accuracy and computational complexity depend on the basis size, which grows with both the carrier frequency and the aperture size~\cite{Beamforming_Optimization, Optimal_Beamforming, Continuous_aperture_arrays}. In contrast, direct functional methods apply variational analysis in the continuous domain, thereby avoiding the approximation loss~\cite{Multi_Group_Multicast, Beamforming_Design, Optimal_Waveform_Design}. Nevertheless, the involved iterative functional updates and repeated numerical integration still incur substantial computational cost. This tradeoff between accuracy and complexity motivates the development of learning-based methods for optimizing continuous beamforming functions with fast inference.

\subsection{Related Works}
Deep neural networks (DNNs) have emerged as a promising approach to reduce the computational cost of CAPA beamforming. However, standard DNNs map finite-dimensional inputs to finite-dimensional outputs, whereas CAPA channels and beamformers are continuous functions. To address the mismatch, in a multiuser single-CAPA downlink, where a CAPA-equipped base station (BS) serves multiple single-antenna users, the optimal beamforming function is simplified to a weighted sum of channel response functions. This allows a network to learn only a finite set of weights~\cite{Multi_User_CAPA, Transformer_Enabled}. However, this method does not extend to multi-CAPA systems where user-side CAPAs cause these weights to become functions. 
For single-user multi-CAPA systems, an implicit neural representation (INR) was proposed to parameterize the continuous function~\cite{Implicit_Neural_CAPA}. This method, however, employs a fully connected neural network (FNN), which suffers from high training complexity and lacks generalizability, e.g., requiring retraining when the number of users~changes.

To reduce training complexity and enhance size generalizability of DNNs, existing studies have exploited mathematical properties of target policies, such as permutation equivariance (PE), for DNN design. Various graph neural networks (GNNs) have been designed to leverage different PE properties for wireless communication tasks, including power allocation~\cite{Optimal_Wireless_Resource_Allocation, Learning_power, GNNs_for_Scalable_Radio, Learn_to_optimize}, user scheduling~\cite{Graph_Embedding_Based, Learning_Based_User, Learning_Wideband_User}, and conventional beamforming~\cite{Understanding_the_Performance, Designing_Heterogeneous_GNNs, Learning_Hybrid_Precoding}.  

Another effective technique is model-driven learning, where mathematical models are incorporated into the DNN design~\cite{Iterative_Algorithm_Induced, Deep_Graph_Unfolding, Learn_to_Rapidly, Transfer_Learning_and, A_Bipartite_Graph, Joint_User_Scheduling, A_Model_based_DNN, A_Model_Based, A_Gradient_Driven, Gradient_Driven_Graph, Gradient_GNN}.
A representative example is deep unfolding, which maps the iterations of a traditional optimization algorithm into trainable network layers. For example, this has been used to unfold the weighted minimum mean-squared error (WMMSE) algorithm to mitigate the complexity of matrix inversions~\cite{Iterative_Algorithm_Induced, Deep_Graph_Unfolding}, and the projected-gradient descent to learn adaptive step sizes~\cite{Learn_to_Rapidly}. While these networks offer better interpretability, they often retain many of the numerical operations of the original algorithms and can remain computationally intensive. Other model-driven learning includes exploiting the structure of an optimal solution to simplify the learning problem~\cite{Transfer_Learning_and, A_Bipartite_Graph, Joint_User_Scheduling, A_Model_based_DNN} or using mathematical models to guide the design of GNNs. For example, in~\cite{A_Model_Based} 
a model-driven GNN was constructed based on the Taylor expansion of the matrix pseudo-inverse, and in~\cite{A_Gradient_Driven, Gradient_Driven_Graph, Gradient_GNN} gradient-driven GNNs were designed by exploiting the similarity between GNN updates and gradient descent. 

These GNN-based and model-driven learning methods were developed for SPDA systems. Consequently, they cannot be directly applied to CAPA systems, where the key challenge remains the learning of continuous beamforming functions.

\subsection{Motivation and Contributions}
While INR-based methods can learn continuous beamforming functions, their application has been limited to single-user multi-CAPA systems. Moreover, these methods have not exploited PE properties or mathematical models to guide the DNN design, leading to high training complexity and poor generalization. The extension to the more complex multiuser multi-CAPA scenario is nontrivial, as its achievable sum rate remains unknown and the structure of the optimal solution has not been characterized to facilitate the model-driven~design.

This paper aims to address these limitations and the main contributions are summarized as follows.\footnote{A part of this work, specifically the derivation of the functional WMMSE algorithm, was reported in a conference paper~\cite{Functional_WMMSE}. This manuscript substantially extends~\cite{Functional_WMMSE} by deriving a closed-form expression for the achievable sum rate in multiuser multi-CAPA systems, designing BeamINR to learn the continuous beamforming function, and providing additional simulation results for both the numerical optimization and the INR-based~methods.}

\begin{itemize}
\item We derive a novel closed-form expression for the achievable sum rate in a multiuser multi-CAPA system, where both the BS and users are equipped with CAPAs. This expression provides the foundation for the numerical and learning-based beamforming methods \mbox{developed~subsequently}.
\item We develop a functional WMMSE algorithm for CAPA systems by first reformulating the sum-rate maximization problem as an equivalent sum mean-squared error (MSE) minimization problem and then employing an orthonormal basis expansion to convert the functional optimization into a parameter optimization problem. By mapping the optimality conditions of this parameter optimization back to the continuous function domain, we derive the iterative update equations for the functional WMMSE.
\item We propose BeamINR, a GNN-based INR that learns the continuous beamforming function. BeamINR is designed to exploit the PE property of the optimal beamforming policy, and its architecture is inspired by the iterative structure of our functional WMMSE. Simulation results show that BeamINR achieves a sum rate close to the functional WMMSE with substantially lower inference latency. Compared with INR-based baselines, BeamINR reduces training complexity and generalizes better to the number of users, CAPA size, and carrier~frequency. 
\end{itemize}  

\textit{Notations:} $(\cdot)^{\mathsf T}$, $(\cdot)^{\mathsf H}$, $(\cdot)^{*}$, and $\|\cdot\|$ denote the transpose, Hermitian transpose, conjugate, and Frobenius norm of a matrix, respectively. $|\cdot|$ denotes the magnitude of a complex value, and $\mathbf{I}_z$ is the identity matrix of size $z\times z$.

\section{System Model}
\begin{figure}%[htbp]
\centering
\includegraphics[width=0.4\textwidth]{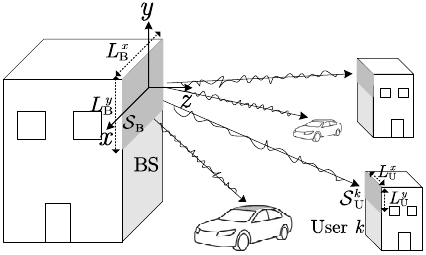}
\caption{Illustration of a multiuser multi-CAPA system.}  \label{CAPA}
\end{figure}
Consider a downlink multiuser multi-CAPA system shown in Fig.~\ref{CAPA}, where a BS equipped with a CAPA serves $K$ users, each also equipped with a CAPA. The rectangular BS aperture is placed on the $xy$-plane and centered at the origin, with side lengths $L_{\mathrm{B}}^x$ and $L_{\mathrm{B}}^y$ along the $x$- and $y$-axes, respectively. A point on the BS aperture is denoted by $\mathbf{s}=[s_x, s_y, 0]^{\mathsf{T}}$, and the set of all such points is denoted by $\mathcal{S}_\mathrm{B}$. The $k$-th user CAPA is centered at location $\mathbf{r}^k_{o}$ and has side lengths $L_{\mathrm{U}}^{x}$ and $L_{\mathrm{U}}^{y}$. 
Since the user CAPA may have an arbitrary spatial orientation, we describe its rotations around the three coordinate axes by angles $\omega^k_x$, $\omega^k_y$, and $\omega^k_z$, with the corresponding rotation matrices $\mathbf{R}_x(\omega^k_x)$, $\mathbf{R}_y(\omega^k_y)$, and $\mathbf{R}_z(\omega^k_z)$. To describe points on the $k$-th user CAPA, we introduce a local coordinate system, where the origin is $\mathbf{r}^k_o$ and the $xy$-plane coincides with the user CAPA. In this local coordinate system, a point on the $k$-th user CAPA is denoted by $\bar{\mathbf{r}}^k=[\bar{r}^k_x, \bar{r}^k_y, 0]^{\mathsf{T}}$, and the set of all such points is denoted by $\bar{\mathcal{S}}^k_\mathrm{U}$. The corresponding point in the global coordinate system is
\begin{equation} \label{rotation equation}
    \mathbf{r} =\mathbf{R}_x(\omega^k_x)\mathbf{R}_y(\omega^k_y)\mathbf{R}_z(\omega^k_z)\bar{\mathbf{r}}^k + \mathbf{r}^k_o,\quad \bar{\mathbf{r}}^k\in\bar{\mathcal{S}}^k_{\mathrm{U}},
\end{equation}
where $\mathcal{S}^k_\mathrm{U}$ denotes the region occupied by the $k$-th user CAPA in the global coordinate system.

The BS transmits $d$ data streams to each user. Let $\mathbf{v}_k(\mathbf{s})=[v_{1}^k(\mathbf{s}), \ldots, v_{d}^k(\mathbf{s})]\in\mathbb{C}^{1\times d}$ denote the transmit beamforming function for user $k$, where $v_{i}^k(\mathbf{s})$ specifies the continuous current distribution associated with the $i$-th stream. Let $\mathbf{x}_k=[x_{1}^k, \ldots, x_{d}^k]^{\mathsf T}\in\mathbb{C}^{d\times 1}$ denote the data symbols intended for user $k$, with $\mathbb{E}\{|x_i^k|^2\} = 1$. The received signal at point $\mathbf{r}$ on the $k$-th user CAPA is
\begin{equation}\label{receive yk}
\!\!{y}_k(\mathbf{r})\!= \sum_{i=1}^K{\int_{\mathcal{S}_{\mathrm{B}}}{\!h_k\left( \mathbf{r},\mathbf{s} \right) \mathbf{v}_i(\mathbf{s})\mathbf{x}_i\mathrm{d}\mathbf{s}}}+n_k(\mathbf{r}),\, \mathbf{r}\in \mathcal{S}_{\mathrm{U}}^k,
\end{equation}
where $h_k(\mathbf{r}, \mathbf{s})$ is the continuous channel kernel\footnote{In functional analysis, a two-variable function used in an integral operator, such as $h_k(\mathbf{r}, \mathbf{s})$ in~\eqref{receive yk}, is referred to as the kernel of the integral operator.} from point $\mathbf{s}$ on the BS CAPA to point $\mathbf{r}$ on the $k$-th user CAPA, and $n_k(\mathbf{r}) \sim \mathcal{CN}(0, \sigma_n^2)$ is additive white Gaussian noise with variance $\sigma_n^2$.

As widely used in the literature~\cite{Beamforming_Design, Multi_User_CAPA}, we consider single-polarized CAPAs under line-of-sight conditions. The BS is polarized along the $y$-axis, while each user's polarization rotates with its local coordinate system. The channel response function $h_k(\mathbf{r}, \mathbf{s})$ is expressed~as
\begin{equation} \label{Green's function}
\!\!\!h_k(\mathbf{r}, \mathbf{s})\!=\!\bm{\Lambda}_{\mathrm{R}}^{k,\mathsf T}\frac{-j \eta e^{-j \frac{2\pi}{\lambda} \| \mathbf{r} - \mathbf{s} \|}}{2\lambda \| \mathbf{r}\! -\! \mathbf{s} \|} 
\!\bigg(\mathbf{I}_3\!-\!\frac{(\mathbf{r}\! -\! \mathbf{s})(\mathbf{r} \!- \!\mathbf{s})^{\mathsf T}}{\| \mathbf{r} - \mathbf{s} \|^2} \bigg)\!\bm{\Lambda}_{\mathrm{T}},
\end{equation}
where $\mathbf{r}\in\mathcal{S}_{\mathrm{U}}^k$, $\mathbf{s}\in \mathcal{S}_{\mathrm{B}}$, $\bm{\Lambda}_{\mathrm{T}} = [0, 1, 0]^{\mathsf T}$ and $\bm{\Lambda}^k_{\mathrm{R}} = \mathbf{R}_x(\omega_x^k)\mathbf{R}_y(\omega_y^k)\mathbf{R}_z(\omega_z^k)\bm{\Lambda}_{\mathrm{T}}$ are unit polarization vectors, $\eta$ is the intrinsic impedance, and $\lambda$ is the wavelength. To focus on beamforming optimization, we assume perfect channel state information at the BS. To obtain the continuous channel, parametric methods can estimate a finite set of channel parameters and reconstruct the channel response functions~\cite{Parametric_Channel_Estimation, Fourier_Plane_Wave}.

\section{Problem Formulation}\label{Sec:Problem Formulation}
We aim to maximize the sum rate by optimizing the continuous beamforming functions. While prior work has derived closed-form expressions for the achievable rate in multiuser single-CAPA systems~\cite{Pattern_Division, Optimal_Beamforming} and single-user multi-CAPA systems~\cite{Beamforming_Design}, these scenarios only account for either inter-user or intra-user interference, respectively. The multiuser multi-CAPA system considered in this paper, however, involves both types of interference simultaneously. We therefore begin by deriving a closed-form expression for its achievable sum rate to formulate the corresponding optimization problem. To this end, we first need to define the inverse of a continuous~kernel.

\begin{definition}[Inverse of a Continuous Kernel]\label{Definition:Inverse of continuous kernel}
For a continuous kernel $G(\mathbf r, \mathbf s)$, a kernel $G^{-1}(\mathbf z, \mathbf r)$ is defined as the inverse of $G(\mathbf r, \mathbf s)$ if it satisfies~\cite{Optimal_Beamforming}
\begin{equation}\label{eq:inverse_of_kernel}
 \int_{\mathcal S} G^{-1}(\mathbf z,\mathbf r)G(\mathbf r,\mathbf s)\mathrm d\mathbf r = \delta(\mathbf z-\mathbf s),
\end{equation}
for all $\mathbf{r}, \mathbf{s}, \mathbf{z}\in\mathcal{S}$, where $\delta(\cdot)$ is the Dirac delta function.
\end{definition}

\begin{proposition}\label{proposition_capacity}
The achievable sum rate of the multiuser multi-CAPA system is given by 
\begin{equation}\label{eq:theorem1 capacity}
    R= \sum\limits_{k=1}^{K} \log\det\big(\mathbf{I}_{d} +\mathbf{Q}_k\big)
\end{equation}
where 
\begin{subequations}\label{proposition:WSR-CAPA}
\begin{align}
&   \mathbf{Q}_k\!=\!
 \iint_{\mathcal{S}_{\mathrm{U}}}\!\!
\mathbf{a}_{kk}^{\mathsf{H}}(\mathbf{r}_1)\,
\mathrm{J}_{\bar{k}}^{-1}(\mathbf{r}_1,\mathbf{r}_2)\,
\mathbf{a}_{kk}(\mathbf{r}_2)\,
\mathrm{d}\mathbf{r}_2\,\mathrm{d}\mathbf{r}_1,\label{eq:theorem1 Q}\\
&   \mathrm{J}_{\bar{k}}(\mathbf{r}_1,\mathbf{r}_2)\!=\! \sum\limits_{j=1,j\neq k}^K\!\! \mathbf{a}_{kj}(\mathbf{r}_1)\,\mathbf{a}_{kj}^{\mathsf{H}}(\mathbf{r}_2) 
\!+\!{\sigma_n^{2}}\delta(\mathbf{r}_1-\mathbf{r}_2), \label{eq:theorem1 R}\\
&    \mathbf{a}_{kj}(\mathbf{r}) 
=  \int_{\mathcal{S}_{\mathrm{B}}}\!h_k(\mathbf{r},\mathbf{s})\,\mathbf{v}_j(\mathbf{s})\,\mathrm{d}\mathbf{s}.\label{eq:theorem1 a}
\end{align}
\end{subequations}
Here, $\mathcal{S}_{\mathrm{U}}=\bigcup_{k=1}^K{\mathcal{S}_{\mathrm{U}}^k}$ and $h_k(\mathbf{r},\mathbf{s})$ is defined to be zero for $\mathbf{r}\notin\mathcal{S}_{\mathrm{U}}^k$.
\end{proposition}
\begin{IEEEproof}
See Appendix~\ref{appendix_capacity}.
\end{IEEEproof}

With Proposition~\ref{proposition_capacity}, the sum-rate maximization problem is formulated as
\begin{subequations}\label{P1:WSR-CAPA}
\begin{align}
\max_{\mathbf{v}_k(\mathbf{s})} 
&   \sum\limits_{k=1}^{K} \log\det\big(\mathbf{I}_{d} +\mathbf{Q}_k\big),\label{P1:capacity} \\
\text{s.t.}\,\, &   \sum\limits_{k=1}^K\int_{\mathcal{S}_{\mathrm{B}}}{\| \mathbf{v}_k(\mathbf{s}) \|^2} \mathrm{d}\mathbf{s}\le C_{\max},\label{P1:power constraint}\\
&\eqref{eq:theorem1 Q},\eqref{eq:theorem1 R},\eqref{eq:theorem1 a},\nonumber
\end{align}
\end{subequations}
where $C_{\max}$ denotes the maximum current budget at the BS.

\section{Functional WMMSE Algorithm}
In this section, we develop a functional WMMSE algorithm for beamforming design. Instead of solving~\eqref{P1:WSR-CAPA} directly, we first reformulate it as an equivalent sum-MSE minimization problem and then represent the continuous functions in the equivalent problem with complete orthonormal basis expansions. These reformulations convert the original functional sum-rate maximization problem into an equivalent optimization over coefficient matrices, facilitating the derivation of the proposed algorithm.

\subsection{Equivalent MSE Minimization Problem}
We first reformulate problem~\eqref{P1:WSR-CAPA} as an equivalent sum-MSE minimization problem. Specifically, by applying the combining function $\mathbf{u}_k(\mathbf{r})\in \mathbb{C} ^{1\times d}$ to the received signal, user $k$ estimates its transmitted signal as
\begin{equation} \label{eq:estimated signal}
\hat{\mathbf{x}}_k= \int_{\mathcal{S}^k _{\mathrm{U}}}{\mathbf{u}_{k}^{\mathsf H}\left( \mathbf{r} \right) y_k\left( \mathbf{r} \right)}\mathrm{d}\mathbf{r}.
\end{equation}
With~\eqref{eq:estimated signal}, the MSE matrix $\mathbf{E}_k$ can be derived as
\begin{equation}\label{eq:MSE matrix}
\begin{split}
\raisetag{11.0ex}
&\mathbf{E}_k=\mathbb{E}_{\mathbf{x},n}\!\left[
 (\hat{\mathbf{x}}_k-\mathbf{x}_k)(\hat{\mathbf{x}}_k-\mathbf{x}_k)^{\mathsf H}\right]
\\
&\triangleq\mathbf{I}_d\!-\!\mathbf{B}_{kk}\!-\!\mathbf{B}_{kk}^{\mathsf H}\!+  \sum_{j=1}^K\mathbf{B}_{kj}\mathbf{B}_{kj}^{\mathsf H}\!+\!\sigma_n^{2}\int_{\mathcal{S} _{\mathrm{U}}^{k}}\!\mathbf{u}_k^{\mathsf H}(\mathbf r)\mathbf{u}_k(\mathbf r)\,\mathrm d\mathbf r,
\end{split}
\end{equation}
where
\begin{equation}\label{eq:equation of B}
    \mathbf{B}_{kj}=  \int_{\mathcal{S} _{\mathrm{U}}^{k}}{\mathbf{u}_{k}^{\mathsf H}\left( \mathbf{r} \right) \mathbf{a}_{kj}(\mathbf{r})}\mathrm{d}\mathbf{r}.
\end{equation}
\begin{proposition} \label{proposition:MSE}
Problem~\eqref{P1:WSR-CAPA} is equivalent to the following problem, in the sense that the globally optimal solution $\mathbf{v}_k(\mathbf{s})$ is identical for both problems.
\begin{subequations}\label{P2:MMSE}
\begin{align}
\!\!&\min_{\{\mathbf{W}_k,\mathbf{u}_k(\mathbf{r}),\mathbf{v}_k(\mathbf{s})\}}   \sum\limits _{k=1}^K{ \mathrm{Tr}\left( \mathbf{W}_k\mathbf{E}_k \right) -\log\det \!\left( \mathbf{W}_k \right)}\label{P2:object} \\
&\quad\quad\,\,\,\,\,\mathrm{s.t.}\quad
 \eqref{eq:MSE matrix},\eqref{eq:equation of B},\eqref{eq:theorem1 a}, \eqref{P1:power constraint},\nonumber
\end{align}
\end{subequations}
where $\mathbf{W}_k\succeq\mathbf{0}$ is the weight matrix of user $k$.
\end{proposition}
\begin{IEEEproof}
See Appendix~\ref{appendix:MSE}.
\end{IEEEproof}

For the beamforming functions $\mathbf{v}_k\left( \mathbf{s} \right)$, $\forall k$, let $\bm{\upbeta}\left( \mathbf{s} \right)\!=\!\left[ \beta_1\left( \mathbf{s} \right) ,\ldots ,\beta_{N_s}\left( \mathbf{s} \right) \right]\!\in\! \mathbb{C} ^{1\times N_s}$ denote orthonormal basis functions with $N_s\!\to\!\infty$, which satisfy the orthonormality~condition
\begin{equation} \label{eq:beta unity}
     \int_{\mathcal{S} _{\mathrm{B}}}{\bm{\upbeta }^{\mathsf{H}}\left( \mathbf{s} \right) \bm{\upbeta }\left( \mathbf{s} \right)}\mathrm{d}\mathbf{s}=\mathbf{I}_{N_s}.
\end{equation}
Then, each beamforming function $\mathbf{v}_k\left( \mathbf{s} \right)$ can be expressed as a linear combination of these basis functions~\cite{Fourier_series}, i.e.,
\begin{equation} \label{eq:v function eq beta V}
    \mathbf{v}_k\left( \mathbf{s} \right) =\bm{\upbeta }\left( \mathbf{s} \right) \mathbf{V}_k,
\end{equation}
where $\mathbf{V}_k\in \mathbb{C} ^{N_s\times d}$ denotes the coefficient matrix, obtained through the following projection
\begin{equation} \label{eq:V eq beat v function}
\mathbf{V}_k= \int_{\mathcal{S} _{\mathrm{B}}}{\bm{\upbeta }^{\mathsf{H}}\left( \mathbf{s} \right) \mathbf{v}_k\left( \mathbf{s} \right)}\mathrm{d}\mathbf{s}.
\end{equation}

Similarly, for $\mathbf{u}_k\left( \mathbf{r} \right)$, consider another set of orthonormal basis functions denoted by $\bm{\upalpha }_k\left( \mathbf{r} \right) =\left[ \upalpha _{1}^{k}\left( \mathbf{r} \right) ,\ldots ,\upalpha _{N_r}^{k}\left( \mathbf{r} \right) \right] \in \mathbb{C} ^{1\times N_r}$ with $N_r\!\to\!\infty$, which satisfy
\begin{equation} \label{eq:alpha unity}
     \int_{\mathcal{S} _{\mathrm{U}}^{k}}{\bm{\upalpha }_{k}^{\mathsf H}\left( \mathbf{r} \right) \bm{\upalpha }_{k}\left( \mathbf{r} \right)}\mathrm{d}\mathbf{r}=\mathbf{I}_{N_r}.
\end{equation}
Accordingly, $\mathbf{u}_k\left( \mathbf{r} \right)$ can be represented as
\begin{equation} \label{eq:u function eq U}
    \mathbf{u}_k\left( \mathbf{r} \right) =\bm{\upalpha }_k\left( \mathbf{r} \right) \mathbf{U}_k,
\end{equation}
where $\mathbf{U}_k\in \mathbb{C} ^{N_r\times d}$ is the corresponding coefficient matrix, expressed as
\begin{equation} \label{eq:U eq u function}
\mathbf{U}_k= \int_{\mathcal{S}^k _{\mathrm{U}}}{\bm{\upalpha }_{k}^{\mathsf H}\left( \mathbf{r} \right) \mathbf{u}_k\left( \mathbf{r} \right)}\mathrm{d}\mathbf{r}.
\end{equation}

Since $\bm{\upbeta }\left( \mathbf{s} \right)$ and $\bm{\upalpha }_k\left( \mathbf{r} \right)$ form complete orthonormal sets over $\mathcal{S}_{\mathrm{B}}$ and $\mathcal{S} _{\mathrm{U}}^{k}$, respectively, the continuous channel kernel $h_k\left( \mathbf{r},\mathbf{s} \right)$ can be represented as~\cite{Beamforming_Design}
\begin{equation} \label{eq:h function eq ahb}
    h_k\left( \mathbf{r},\mathbf{s} \right) =\bm{\upalpha }_k\left( \mathbf{r} \right) \mathbf{H}_k\bm{\upbeta }^{\mathsf{H}}\left( \mathbf{s} \right), 
\end{equation}
where $\mathbf{H}_k\!\in\! \mathbb{C} ^{N_r\times N_s}$ is the channel coefficient matrix, expressed as
\begin{equation}\label{eq:H eq h function}
\mathbf{H}_k= \int_{\mathcal{S} _{\mathrm{U}}^{k}}{\int_{\mathcal{S} _{\mathrm{B}}}{\bm{\upalpha }_{k}^{\mathsf H}\left( \mathbf{r} \right) h_k\left( \mathbf{r},\mathbf{s} \right) \bm{\upbeta }\left( \mathbf{s} \right)}}\mathrm{d}\mathbf{s}\mathrm{d}\mathbf{r}.
\end{equation}

Substituting~\eqref{eq:v function eq beta V},~\eqref{eq:u function eq U}, and~\eqref{eq:h function eq ahb} into problem~\eqref{P2:MMSE} yields
\begin{subequations}\label{P4:discret MSE}
\begin{align}
\!\! &\min_{\mathbf{W}_k,\mathbf{U}_k,\mathbf{V}_k} \sum_{k=1}^K{\left( \mathrm{Tr}\left( \mathbf{W}_k\mathbf{E}_k \right) -\log\det \!\left( \mathbf{W}_k \right) \right)}\label{P4:objective}\\
&\,\,\,\,\,\,\,\,\,\,\text{s.t.}\quad\mathbf{E}_k=\mathbf{I}_{d}\!-\!\mathbf{B}_{kk}\!-\!\mathbf{B}_{kk}^{\mathsf{H}}\nonumber\\
&\quad\quad\quad\quad\quad\quad\quad+\sum\limits_{j=1}^K{\mathbf{B}_{kj}\mathbf{B}_{kj}^{\mathsf{H}}}\!+\sigma_n^{2}\mathbf{U}_{k}^{\mathsf{H}}{\mathbf{U}}_k,\label{eq:MSE matrix matrix}\\
&\,\,\,\,\,\,\,\,\,\,\quad\quad\,\,\mathbf{B}_{kj}=\mathbf{U}_{k}^{\mathsf{H}}\,\mathbf{H}_k\,\mathbf{V}_j,\label{eq:equation of B matrix}\\
&\,\,\,\,\,\,\,\,\,\,\quad\quad\,\, \sum_{k=1}^K{\mathrm{Tr}\left( \mathbf{V}_{k}^{\mathsf{H}}\mathbf{V}_k \right)}\le C_{\max}.
\end{align}
\end{subequations}
Problem~\eqref{P4:discret MSE} optimizes the coefficient matrices $\mathbf{V}_k$ and $\mathbf{U}_k$ together with the weight matrices $\mathbf{W}_k$. Although this problem can be solved by the conventional WMMSE algorithm, the computational complexity becomes prohibitive as $N_s \!\to\!\infty$ and $N_r \!\to\!\infty$. Therefore, instead of directly optimizing $\mathbf{V}_k$ and $\mathbf{U}_k$, the next subsection derives their optimality conditions and uses them to construct the corresponding optimal functions $\mathbf{v}_k(\mathbf{s})$ and $\mathbf{u}_k(\mathbf{r})$.

\subsection{Derivation of the Functional WMMSE Algorithm}
This subsection derives the update equations for $\mathbf{u}_k(\mathbf{r})$, $\mathbf{W}_k$, and $\mathbf{v}_k(\mathbf{s})$.

\subsubsection{Update of $\mathbf{u}_k(\mathbf{r})$}
From problem~\eqref{P4:discret MSE}, the first-order optimality condition for $\mathbf{U}_k$ is
\begin{equation} \label{eq:optimal condition of U}
\bigg( \sum_{j=1}^K{\mathbf{H}_k\mathbf{V}_j\mathbf{V}_{j}^{\mathsf{H}}\mathbf{H}_{k}^{\mathsf{H}}+\sigma_{n}^{2}\mathbf{I}_{N_r}} \bigg) \mathbf{U}_k-\mathbf{H}_k\mathbf{V}_k=0.
\end{equation}
Multiplying both sides of~\eqref{eq:optimal condition of U} by $\bm{\upalpha }_{k}\left( \mathbf{r}_1 \right)$ yields 
\begin{equation}\label{eq:optimal condition of U1}
\!\!\bm{\upalpha}_{k}\!\left(\! \mathbf{r}_1 \!\right) \mathbf{H}_k\!\mathbf{V}_k\!=\!\bm{\upalpha}_{k}\!\left(\!\mathbf{r}_1\!\right)\!\bigg(\!  \sum_{j=1}^K\!{\mathbf{H}_k\mathbf{V}_j\mathbf{V}_{j}^{\mathsf{H}}\mathbf{H}_{k}^{\mathsf{H}}\!+\!\sigma_n^{2} \mathbf{I}_{N_r}} \!\bigg)\!\mathbf{U}_k.
\end{equation}
Using the orthonormality condition in~\eqref{eq:beta unity}, the term $\bm{\upalpha}_{k}\!\left( \mathbf{r}_1 \right) \mathbf{H}_k\!\mathbf{V}_k$ can be rewritten as 
\begin{equation} \label{eq:a function eq ahv}
\begin{aligned}
\bm{\upalpha }_{k}\left( \mathbf{r}_1 \right) \mathbf{H}_k\mathbf{V}_j&= \int_{\mathcal{S} _{\mathrm{B}}}\underset{h_k\left( \mathbf{r}_1,\mathbf{s} \right)}{\underbrace{\bm{\upalpha }_{k}\left( \mathbf{r}_1 \right) \mathbf{H}_k{\bm{\upbeta }\left( \mathbf{s} \right)^{\mathsf{H}}} }}\underset{\mathbf{v}_j\left( \mathbf{s} \right)}{\underbrace{\bm{\upbeta }\left( \mathbf{s} \right)\mathbf{V}_j}}\mathrm{d}\mathbf{s}\\
&= \int_{\mathcal{S} _{\mathrm{B}}}{h_k\left( \mathbf{r}_1,\mathbf{s} \right) \mathbf{v}_j\left( \mathbf{s} \right)}\mathrm{d}\mathbf{s}\triangleq\mathbf{a}_{kj}\left( \mathbf{r}_1 \right),
\end{aligned}
\end{equation}
where the second equality follows from~\eqref{eq:h function eq ahb} and~\eqref{eq:v function eq beta V}. Substituting~\eqref{eq:a function eq ahv} into~\eqref{eq:optimal condition of U1} gives
\begin{equation}
\begin{split}\label{eq:optimal condition of U2}
%\raisetag{3.0ex}
\mathbf{a}_{kk}(\mathbf r_1)\!=\!\!\bigg(\! \sum_{j=1}^{K}\mathbf{a}_{kj}(\mathbf r_1)\mathbf V_j^{\mathsf H}\mathbf H_k^{\mathsf H}\!+\!\sigma_n^{2} \bm{\upalpha}_{k}(\mathbf r_1)\bigg)\mathbf I_{N_r}\mathbf U_k.
\end{split}
\end{equation}
With~\eqref{eq:alpha unity}, the right-hand side of~\eqref{eq:optimal condition of U2} can be further expressed~as
\begin{equation}
\begin{split}\label{eq:optimal condition of U3}
\raisetag{5.0ex}
&\int_{\mathcal S_{\mathrm U}^k}\!\!\!\bigg(\!\sum_{j=1}^{K}\mathbf{a}_{kj}(\mathbf r_1)\underset{\mathbf{a}_{kj}^{\mathsf{H}}\!\left( \mathbf{r} \right)}{\underbrace{\mathbf{V}_{j}^{\mathsf{H}}\mathbf{H}_{k}^{\mathsf{H}}\bm{\upalpha }_{k}^{\mathsf{H}}(\mathbf{r})}}\!+\!\sigma_n^{2} \bm{\upalpha}_{k}(\mathbf r_1) \bm{\upalpha}_{k}^{\mathsf H}(\mathbf r)\bigg)
 \bm{\upalpha}_{k}(\mathbf r)\mathbf U_k\mathrm d\mathbf r\\
&= \int_{\mathcal{S}_{\mathrm{U}}^{k}}{\bigg( \sum_{j=1}^K{\mathbf{a}_{kj}\left(\mathbf{r}_1\right) \mathbf{a}_{kj}^{\mathsf H}\left( \mathbf{r} \right) +\sigma_n^{2}\delta \left( \mathbf{r}_1-\mathbf{r} \right)} \bigg)\!\mathbf{u}_{k}\left( \mathbf{r} \right)}\mathrm{d}\mathbf{r},
\end{split}
\end{equation}
where $ \bm{\upalpha}_{k}(\mathbf r_1) \bm{\upalpha}_{k}^{\mathsf H}(\mathbf r)=\delta \left( \mathbf{r}_1-\mathbf{r} \right)$ holds due to the orthonormality of $ \bm{\upalpha}_{k}(\mathbf r)$, and $ \bm{\upalpha}_{k}(\mathbf r)\mathbf U_k$ equals $\mathbf{u}_{k}\!\left( \mathbf{r} \right)$ by~\eqref{eq:u function eq U}.

Using~\eqref{eq:optimal condition of U3} to replace the right-hand side of~\eqref{eq:optimal condition of U2}, the functional optimality condition for $\mathbf{u}_k(\mathbf{r})$ is
 \begin{equation}
\begin{aligned}\label{eq:functional optimal condition of U}
\!\!\!\mathbf{a}_{kk}(\mathbf r_1)\!\!=\!\! \int_{\mathcal{S}_{\mathrm{U}}^{k}}\!\!{\bigg(\! \sum_{j=1}^K{\!\mathbf{a}_{kj}\!\left(\!\mathbf{r}_1\!\right) \mathbf{a}_{kj}^{\mathsf H}\!\left( \mathbf{r} \right) \!+\!\sigma_n^{2}\delta \left( \mathbf{r}_1\!\!-\!\!\mathbf{r} \right)}\! \bigg)\mathbf{u}_{k}\!\left( \mathbf{r} \right)}\mathrm{d}\mathbf{r}.
\end{aligned}
\end{equation}
Defining $\mathrm{J}_{k}(\mathbf{r}_1,\mathbf{r})\triangleq \sum\limits_{j=1}^K{\mathbf{a}_{kj}\left(\!\mathbf{r}_1\right) \mathbf{a}_{kj}^{\mathsf H}\!\left( \mathbf{r} \right) +\sigma_n^{2}\delta \left( \mathbf{r}_1-\mathbf{r} \right)}$,~\eqref{eq:functional optimal condition of U} can be rewritten as $\mathbf{a}_{kk}(\mathbf r_1)= \int_{\mathcal{S}_{\mathrm{U}}^{k}}\mathrm{J}_{k}(\mathbf{r}_1,\mathbf{r}){\mathbf{u}_{k}\left( \mathbf{r} \right)}\mathrm{d}\mathbf{r}$, from which $\mathbf{u}_{k}\!\left( \mathbf{r} \right)$ is obtained as
\begin{equation}\label{eq:functional update u}
    \mathbf{u}_{k}\left( \mathbf{r} \right) = \int_{\mathcal{S} _{\mathrm{U}}^{k}}{\mathrm{J}_{k}^{-1}\left( \mathbf{r},\mathbf{r}_1 \right)}\mathbf{a}_{kk}\left( \mathbf{r}_1 \right) \mathrm{d}\mathbf{r}_1,
\end{equation}
where $\mathrm{J}_{k}^{-1}\left( \mathbf{r},\mathbf{r}_1 \right)$ denotes the inverse of $\mathrm{J}_{k}\!\left( \mathbf{r}_1,\mathbf{r}_2 \right)$, as defined in Definition~\ref{Definition:Inverse of continuous kernel}.

\subsubsection{Update of $\mathbf{W}_k$}
From problem~\eqref{P4:discret MSE}, the optimal $\mathbf{W}_k$ is
\begin{equation} \label{eq:discrete update W}
\mathbf{W}_k=\left( \mathbf{I}_{d}-\mathbf{U}_{k}^{\mathsf{H}}\mathbf{H}_k\mathbf{V}_k \right) ^{-1}.
\end{equation}
With~\eqref{eq:alpha unity},~\eqref{eq:discrete update W} can be rewritten as 
\begin{equation}
    \begin{aligned}\label{eq:functional of W}
\mathbf{W}_k&={\bigg( \mathbf{I}_d- \int_{\mathcal{S}_{\mathrm{U}}^{k}}\underset{\mathbf{u}_{k}^{\mathsf H}\left( \mathbf{r} \right)}{\underbrace{\mathbf{U}_{k}^{\mathsf H}{{\bm{\upalpha }_{k}^{\mathsf H}\left( \mathbf{r} \right)}}}} \underset{\mathbf{a}_{kk}\left( \mathbf{r} \right)}{\underbrace{\bm{\upalpha }_{k}\left( \mathbf{r} \right)\mathbf{H}_k\mathbf{V}_k}}\mathrm{d}\mathbf{r} \bigg) ^{-1}},
\\
&={\bigg( \mathbf{I}_d- \int_{\mathcal{S} _{\mathrm{U}}^{k}}\mathbf{u}_{k}^{\mathsf H}\left( \mathbf{r} \right) \mathbf{a}_{kk}\left( \mathbf{r} \right)\mathrm{d}\mathbf{r} \bigg) ^{-1}}, 
    \end{aligned}
\end{equation}
where the second equality follows from~\eqref{eq:u function eq U} and~\eqref{eq:a function eq ahv}.

\subsubsection{Update of $\mathbf{v}_k(\mathbf{s})$}
Similar to the derivation of ${\mathbf u}_k(\mathbf r)$, the first-order optimality condition for $\mathbf{V}_k$ can be derived as
\begin{equation}\label{eq:optimal condition of V}
   \bigg(\!  \sum_{j=1}^K{\mathbf{H}_{j}^{\mathsf H}\mathbf{U}_j\mathbf{W}_j\mathbf{U}_{j}^{\mathsf H}\mathbf{H}_j}+\mu\mathbf{I} \bigg) \mathbf{V}_k-\mathbf{H}_{k}^{\mathsf H}\mathbf{U}_k\mathbf{W}_k=0,
\end{equation}
where $\mu$ is the Lagrange multiplier introduced to satisfy the constraint in~\eqref{P1:power constraint}, which can be obtained via bisection. 
Multiplying both sides of~\eqref{eq:optimal condition of V} by $\bm{\upbeta }\left( \mathbf{s}_1 \right)$ yields 
\begin{equation}
\begin{aligned}\label{eq:optimal condition of V1}
\bigg(\sum_{j=1}^K\bm{\upbeta }\left( \mathbf{s}_1 \right){\mathbf{H}_{j}^{\mathsf H}\mathbf{U}_j\mathbf{W}_j\mathbf{U}_{j}^{\mathsf H}\mathbf{H}_j}+&\mu\bm{\upbeta }_{k}\left( \mathbf{s}_1 \right)\mathbf{I} \bigg) \mathbf{V}_k\\
&\quad\quad=\bm{\upbeta }\left( \mathbf{s}_1 \right)\mathbf{H}_{k}^{\mathsf H}\mathbf{U}_k.
\end{aligned}
\end{equation}
Using the orthonormality condition in~\eqref{eq:alpha unity}, the term $\bm{\upbeta }\left( \mathbf{s}_1 \right) \mathbf{H}_{k}^{\mathsf H}\mathbf{U}_k$ can be rewritten as
\begin{equation}
\begin{aligned}\label{eq:c function}
\bm{\upbeta }\left( \mathbf{s}_1 \right) \mathbf{H}_{k}^{\mathsf H}\mathbf{U}_k&= \int_{\mathcal{S}_{\mathrm{U}}^{k}}\underset{h_{k}^{\mathsf H}\left( \mathbf{r},\mathbf{s}_1 \right)}{\underbrace{\bm{\upbeta }\left( \mathbf{s}_1 \right) \mathbf{H}_k^{\mathsf H}{{\bm{\upalpha }_{k}^{\mathsf H}\left( \mathbf{r} \right)} }} }\underset{\mathbf{u}_k\left( \mathbf{r} \right)}{\underbrace{\bm{\upalpha }_{k}\left( \mathbf{r} \right)\mathbf{U}_k}}\mathrm{d}\mathbf{r}
\\
&= \int_{\mathcal{S} _{\mathrm{U}}^{k}}{h_{k}^{\mathsf H}\left( \mathbf{r},\mathbf{s}_1 \right) \mathbf{u}_k\left( \mathbf{r} \right)}\mathrm{d}\mathbf{r}\triangleq\mathbf{c}_{k}\left( \mathbf{s}_1 \right).
\end{aligned}
\end{equation}
where the second equality follows from~\eqref{eq:h function eq ahb} and~\eqref{eq:u function eq U}.
Substituting~\eqref{eq:c function} into~\eqref{eq:optimal condition of V1} yields
\begin{equation}\label{eq:optimal condition of V2}
\mathbf{c}_{k}\left( \mathbf{s}_1 \right) \mathbf{W}_k\!\!=\!\!\bigg(\! \sum\limits_{j=1}^K{\mathbf{c}_{j}\left( \mathbf{s}_1 \right) \mathbf{W}_j\mathbf{U}_{j}^{\mathsf H}\mathbf{H}_j}\!+\!\mu\bm{\upbeta }\left( \mathbf{s}_1 \right) \bigg) \mathbf{V}_k.
\end{equation}
Applying~\eqref{eq:beta unity}, the right-hand side of~\eqref{eq:optimal condition of V2} can be further written as 
\begin{equation}
\label{eq:optimal condition of V right hand}
\begin{split}\raisetag{2.5ex}
&\bigg(\sum_{j=1}^{K}\mathbf c_{j}(\mathbf s_1)\mathbf W_j \mathbf U_j^{\mathsf H}\mathbf H_j +\mu\,\bm\upbeta(\mathbf s_1)\bigg)\mathbf{I}_{N_s}\mathbf V_k\\
&=\int_{\mathcal S_{\mathrm B}}\bigg(\!\sum_{j=1}^{K}\!\mathbf c_{j}(\mathbf s_1)\mathbf W_j\underset{\mathbf c^{\mathsf H}_{j}(\mathbf s)}{\underbrace{\mathbf U_j^{\mathsf H}\mathbf H_j\bm\upbeta^{\mathsf H}(\mathbf s)}}+\\
&\quad\quad\quad\quad\quad\mu\bm\upbeta(\mathbf s_1)\bm\upbeta^{\mathsf H}(\mathbf s)\bigg)
\bm\upbeta(\mathbf s)\mathbf V_k\mathrm d\mathbf s.
\end{split}
\end{equation}
Using~\eqref{eq:optimal condition of V right hand} to replace the right-hand side of~\eqref{eq:optimal condition of V2}, the functional optimality condition for $\mathbf{v}_k(\mathbf{s})$ is
\begin{equation}\label{eq:functional optimal condition of v}
\begin{aligned}
    &\mathbf{c}_{k}\left( \mathbf{s}_1 \right) \mathbf{W}_k=\\
&\,\,\, \int_{\mathcal{S}_{\mathrm{B}}}\!\!\!{\bigg(\! \sum_{j=1}^K{\!\mathbf{c}_{j}\left( \mathbf{s}_1 \right) \mathbf{W}_j\mathbf{c}_{j}^{\mathsf H}\left( \mathbf{s} \right)}\!+\!\mu\delta \left( \mathbf{s}_1\!\!-\!\mathbf{s} \right)\!\bigg)\! \mathbf{v}_k\left( \mathbf{s} \right)}\mathrm{d}\mathbf{s}, 
\end{aligned}
\end{equation}
where $\bm\upbeta(\mathbf s_1)\bm\upbeta^{\mathsf H}(\mathbf s)=\delta \left( \mathbf{s}_1-\mathbf{s} \right)$ holds due to the orthonormality of $\bm{\upbeta}(\mathbf s)$, and $\bm{\upbeta}(\mathbf s)\mathbf V_k$ equals $\mathbf{v}_{k}\left( \mathbf{s} \right)$ based on~\eqref{eq:v function eq beta V}.~Defining $\mathrm{T}_k\left( \mathbf{s}_1,\mathbf{s}\right) 
\!\triangleq\!\sum_{j=1}^K{\!\mathbf{c}_{j}\!\left( \mathbf{s}_1 \right)\!\mathbf{W}_j\mathbf{c}_{j}^{\mathsf H}\left( \mathbf{s} \right)}\!+\!\mu\delta \left( \mathbf{s}_1\!-\!\mathbf{s} \right)$,~\eqref{eq:functional optimal condition of v} can be rewritten as $\mathbf{c}_{k}\left( \mathbf{s}_1 \right) \mathbf{W}_k=\int_{\mathcal{S}_{\mathrm{B}}}{\mathrm{T}_k\left( \mathbf{s}_1,\mathbf{s}\right) \mathbf{v}_k\left( \mathbf{s} \right)}\mathrm{d}\mathbf{s}$, from which $\mathbf{v}_k(\mathbf{s})$ is derived~as
\begin{equation}\label{eq:functional-update-v}
    \mathbf{v}_{k}\left( \mathbf{s} \right) = \!\int_{\mathcal{S} _{\mathrm{B}}}\!\!{\mathrm{T}_{k}^{-1}\!\left( \mathbf{s},\mathbf{s}_1 \right)}\mathbf{c}_{k}\left( \mathbf{s}_1 \right)\mathbf{W}_k \mathrm{d}\mathbf{s}_1,
\end{equation}
where ${\mathrm{T}_{k}^{-1}\!\left( \mathbf{s},\mathbf{s}_1 \right)}$ denotes the inverse of ${\mathrm{T}_{k}\!\left( \mathbf{s}_1,\mathbf{s} \right)}$.

The proposed functional WMMSE algorithm is summarized in Table~\ref{tab:functional-wmmse}.
\begin{remark}
Unlike Fourier-based discretization methods, which iteratively optimize finite-dimensional coefficient matrices and then reconstruct the continuous beamforming functions from truncated Fourier expansions~\cite{Pattern_Division, On_the_SE}, the proposed approach yields closed-form functional update equations and updates the beamforming functions directly in the continuous domain. It therefore avoids the approximation errors introduced by truncating the Fourier expansion.
\end{remark}

\begin{table}[t] 
  \caption{Pseudocode of the Proposed Functional WMMSE Algorithm}
  \label{tab:functional-wmmse}
  \centering
  \small
  \setlength{\tabcolsep}{6pt}
  \renewcommand{\arraystretch}{1.15}
  \begin{tabularx}{0.98\linewidth}
  {|>{\hspace{-0.5\tabcolsep}\centering\arraybackslash}p{1.0em}
    @{\hspace{3pt}}
    X<{\hspace{-0.7\tabcolsep}}|}
    \hline
    1 & \rule{0pt}{1.1em}%
        Initialize $\mathbf{v}_k(\mathbf{s})$ such that $\sum_{k=1}^K\int_{\mathcal{S}_{\mathrm{B}}}{\| \mathbf{v}_k(\mathbf{s}) \|^2} \mathrm{d}\mathbf{s}\le C_{\max}$, and set $ \mathbf W_{k}=\mathbf I_{d}$, $\forall\, k$ \\
    2 & \textbf{repeat} \\
    3 & \quad Update $\mathbf{u}_{k}(\mathbf{r})$ with~\eqref{eq:functional update u}, $\, \forall\, k$ \\
    4 & \quad Update $\mathbf W_{k}$ with~\eqref{eq:functional of W}, $\, \forall\, k$ \\
    5 & \quad Update $\mathbf v_{k}(\mathbf{s})$ with~\eqref{eq:functional-update-v}, $\, \forall\, k$ \\
    6 & \textbf{until}\ the change in $\!\sum\limits_{k=1}^K\!\log\det(\mathbf W_{k})$ is less than a tolerance~$\varepsilon$.%
        \rule[-0.2ex]{0pt}{1.1em} \\[2pt]\hline
  \end{tabularx}
 % \vspace{-2pt}
\end{table}

\section{INR-based Learning of Beamforming Function} \label{Training Method}
In this section, we develop an INR-based method to learn the beamforming policy, i.e., the mapping from the channel response function to the beamforming function. Our design is guided by the PE property of this policy and the iterative structure of the functional WMMSE algorithm. In the following, we first establish the PE property, then implement the INR as a GNN to exploit the PE property, and finally design a model-driven GNN update equation inspired by the WMMSE~iterations.

\subsection{Permutation Equivariant INR with GNNs}\label{BeamINR}
Substituting~\eqref{rotation equation} into~\eqref{eq:theorem1 a}, then into~\eqref{eq:theorem1 Q}, and integrating over $\mathcal{S}_\mathrm{U}$, we can find that $\mathbf{Q}_k$ depends on the matrix $\mathbf{P}_o=[\mathbf{p}_1^{\mathsf{T}}, \cdots, \mathbf{p}_K^{\mathsf{T}}]^{\mathsf{T}}\in\mathbb{R}^{K\times 6}$, where $\mathbf{p}_k=[\mathbf{r}_o^{k,\mathsf{T}},\omega_x^k,\omega_y^k,\omega_z^k]\in\mathbb{R}^{1\times 6}$ collects the location and orientation of user $k$. Consequently, the beamformer $\mathbf{V}(\mathbf{s})=[\mathbf{v}^{\mathsf{T}}_{1}(\mathbf{s}), \ldots, \mathbf{v}^{\mathsf{T}}_{K}(\mathbf{s})]^{\mathsf{T}}\in\mathbb{C}^{K\times d}$ is determined by $\mathbf{P}_o$ and $\mathbf{s}$. We therefore define the optimal beamforming policy as the mapping from $(\mathbf{s},\mathbf{P}_o)$ to the optimal beamforming at point $\mathbf{s}$, which is denoted as
\begin{equation}\label{optimal_policy}
\mathbf{V}^{\star}(\mathbf{s})=\mathcal{F}_b(\mathbf{s}, \mathbf{P}_o),\quad \mathbf{s} \in \mathcal{S}_{\mathrm{B}},
\end{equation}
where $\mathbf{V}^{\star}(\mathbf{s})$ denotes the optimal beamforming at point $\mathbf{s}$ for the input $(\mathbf{s},\mathbf{P}_o)$.

The PE property of a policy means that a permutation of the input indices leads to the same permutation of the output indices, while leaving the objective and constraints of the problem unchanged. Let $\mathcal{H}(\mathbf{r},\mathbf{s})=[h_1(\mathbf{r},\mathbf{s}),\cdots,h_K(\mathbf{r},\mathbf{s})]$ denote the set of channel kernels in problem~\eqref{P1:WSR-CAPA}. 
Permuting the user indices of $\mathbf{P}_o$ induces the same permutation on the indices of $\mathcal{H}(\mathbf{r},\mathbf{s})$. Moreover, when the user indices of \(\mathcal{H}(\mathbf{r},\mathbf{s})\) and \(\mathbf{V}(\mathbf{s})\) are
permuted consistently, the sum rate and the constraints in~\eqref{P1:WSR-CAPA} remain unchanged. Thus, the optimal
policy in~\eqref{optimal_policy} is permutation equivariant with respect to
the user dimension,~i.e.,
\begin{equation}\label{1D-PE}
   \boldsymbol{\Pi}^{\mathsf{T}}\mathbf{V}^{\star}(\mathbf{s})=\mathcal{F}_b\big(\mathbf{s}, \boldsymbol{\Pi}^{\mathsf{T}}\mathbf{P}_{o}\big),
\end{equation}
where $\boldsymbol{\Pi}\!\in\!\mathbb{R}^{K\times K}$ is an arbitrary permutation matrix on the user~indices.

To parameterize the policy in~\eqref{optimal_policy}, we employ an INR given~by
\begin{equation}\label{eq:BeamINR}
   \mathbf{V}(\mathbf{s}) = \mathcal{P}_{\theta}(\mathbf{s}, \mathbf{P}_o),\quad\mathbf{s} \in \mathcal{S}_{\mathrm{B}},
\end{equation}
where $\mathcal{P}_{\theta}(\cdot)$ denotes the INR parameterized by $\theta$. To exploit the PE property in~\eqref{1D-PE}, we implement $\mathcal{P}_{\theta}(\cdot)$ as a GNN on a vertex graph, as illustrated in Fig.~\ref{Undirectional_Graph}. The graph consists of $K$ user vertices with pairwise edges. Vertex $k$ takes the geometry vector $\mathbf{p}_k$ and the spatial coordinate $\mathbf{s}$ as its feature and outputs $\mathbf{v}_k(\mathbf{s})$. No features or outputs are assigned to the edges.
%\vspace{-.5cm}
\begin{figure}%[htbp]
\centering
 \includegraphics[width=0.4\textwidth]{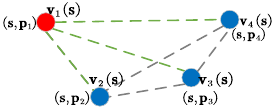}
% \vspace{-0.3cm}
\caption{Illustration of the undirected graph with $K=4$.}  \label{Undirectional_Graph}
\end{figure}

A conventional GNN updates vertex representations by aggregating information from neighboring vertices and combining it with the representation of the target vertex. For vertex $k$, the hidden representation in the $(l+1)$-th layer ${\mathbf{d}}_{k}^{(l+1)}(\mathbf{s}) = [d_{k,1}^{(l+1)}(\mathbf{s}), \dots, d_{k,C_{l+1}}^{(l+1)}(\mathbf{s})]^\mathsf{T}$ is updated as
\begin{equation}\label{Update function of GNN}
   \mathbf{d}_{k}^{(l+1)}(\mathbf{s})\!=\!\sigma \Bigg( \underset{\text{combination}}{\underbrace{\mathbf{S}^{(l)}{\mathbf{d}}_{k}^{(l)}(\mathbf{s})+\overset{\text{aggregation}}{\overbrace{\mathbf{W}^{(l)}\sum_{i=1,i\ne k}^K{\mathbf{d}_{i}^{(l)}}(\mathbf{s})}}}} \Bigg),
\end{equation}	
where $C_{l}$ is the representation dimension in the $l$-th layer, the neighboring representations ${{\mathbf d}}_{i}^{(l)}(\mathbf{s}), \forall i\neq k$, are aggregated through the trainable matrix ${\mathbf W}^{(l)}\in\mathbb{C}^{C_{l+1}\times C_{l}}$, and the resulting aggregated information is combined with the representation of vertex $k$, namely ${\mathbf{d}}_{k}^{(l)}(\mathbf{s})$, through $\mathbf{S}^{(l)}\in\mathbb{C}^{C_{l+1}\times C_{l}}$. $\sigma(\cdot)$ is an element-wise activation function, and the initial and final hidden representations, ${\mathbf{d}}_{k}^{(0)}(\mathbf{s})$ and ${\mathbf{d}}_{k}^{(L)}(\mathbf{s})$, correspond to the vertex features and outputs, respectively, where $L$ is the total number of GNN layers.

In~\eqref{Update function of GNN}, the aggregation and combination processes have the following characteristics.
\begin{itemize}
\item[1)] In the aggregation term, the neighboring hidden representations $\mathbf{d}_{i}^{(l)}(\mathbf{s})$, $\forall i\neq k$, are aggregated using a shared trainable matrix $\mathbf{W}^{(l)}$ that cannot differentiate the contributions of different neighbors.

\item[2)] In the combination term, ${\mathbf{d}}_{k}^{(l)}(\mathbf{s})$ is combined with the aggregated information through $\mathbf{S}^{(l)}$.

\item[3)] The update of $\mathbf{d}_{k}^{(l+1)}(\mathbf{s})$ uses only the information at point $\mathbf{s}$ from layer $l$.
\end{itemize}

Although the conventional GNN update is PE, this property alone does not guarantee generalization across different system sizes. As analyzed in~\cite{Recursive_GNNs} for SPDA beamforming, incorporating mathematical models into the GNN architecture can improve generalization. We therefore analyze the relationship between the GNN update and our functional WMMSE algorithm to develop a model-driven INR layer.

\subsection{Design of the Update Equation} 
To connect the functional WMMSE algorithm to the GNN's update equation, we first express the beamforming iteration from Table~\ref{tab:functional-wmmse} in an explicit recursive form. The following proposition states this recursion, defining each beamforming iterate as a function of the previous one.

\begin{proposition}\label{prop:wmmse_recursion}
Given the beamforming functions $\{\mathbf v_k^{(l)}(\mathbf s)\}_{k=1}^K$ at iteration $l$, the beamforming iteration equation in the proposed functional WMMSE algorithm can be equivalently written as
\begin{equation} \label{eq:iterative equation of vk2}
\begin{split}\raisetag{2.2cm}
&\mathbf{v}_{k}^{\left( l+1 \right)}\!\left( \mathbf{s} \right)
=\int_{\mathcal{S}_{\mathrm{U}}^{k}}h_{k}^{\mathsf{H}}\left( \mathbf{r},\mathbf{s} \right) \mathbf{a}^{(l)}_{kk}\left( \mathbf{r}\right)\mathrm{d}\mathbf{r}\mathbf{\Theta}_{k}\\
&\quad\quad\quad\quad\quad\quad+\sum_{i=1}^K\sum_{\substack{j=1\\ (i,j)\ne (k,k)}}^K
\int_{\mathcal{S} _{\mathrm{U}}^{i}}h_{i}^{\mathsf{H}}\left( \mathbf{r},\mathbf{s} \right) \mathbf{a}^{(l)}_{ij}\left( \mathbf{r} \right)\mathrm{d}\mathbf{r}\mathbf{\Sigma}^k_{ij},
\end{split}
\end{equation}
where $\mathbf{a}^{(l)}_{ij}\!\left( \mathbf{r} \right)\!\!=\!\!\int_{\mathcal{S}_{\mathrm{B}}}
\!\!h_i\left( \mathbf{r},\mathbf{s}_1 \right)\mathbf{v}_{j}^{(l)}\!\!\left( \mathbf{s}_1 \right)\mathrm{d}\mathbf{s}_1$ and $\mathbf{\Theta}_{k}, \mathbf{\Sigma}^k_{ij}\!\in\!\mathbb{C}^{d\times d}$.
\end{proposition}
\begin{IEEEproof}
See Appendix~\ref{appendix:wmmse_recursion}.
\end{IEEEproof}

Both the functional WMMSE recursion~\eqref{eq:iterative equation of vk2} and the GNN update in~\eqref{Update function of GNN} follow a two-stage process, i.e., an aggregation of terms from other users, followed by a combination with a self-term. Specifically, the GNN update~\eqref{Update function of GNN} aggregates the hidden representations of neighboring users, ${\mathbf d_i^{(l)}}(\mathbf s)$, and combines the result with the hidden representation of user $k$, ${\mathbf d_k^{(l)}}(\mathbf s)$. Analogously, the functional WMMSE recursion~\eqref{eq:iterative equation of vk2} aggregates the channel kernels $h_i^{\mathsf H}(\mathbf r,\mathbf s)$ weighted by $\mathbf{a}_{ij}^{(l)}(\mathbf r)$, and combines the result with the self-term $h_k^{\mathsf H}(\mathbf r,\mathbf s)$ weighted by $\mathbf{a}_{kk}^{(l)}(\mathbf r)$.

Despite this structural similarity, the functional WMMSE recursion exhibits four key differences from a conventional GNN update.
\begin{itemize}
\item[1)] It aggregates the channel kernels $h_i^{\mathsf H}(\mathbf r,\mathbf s)$ rather than hidden representations ${\mathbf{d}_i^{(l)}}(\mathbf s)$.
\item[2)] Each user's contribution is weighted by a unique coefficient $\mathbf{a}_{ij}^{(l)}(\mathbf r)$, instead of sharing a common weight across all neighbors.
\item[3)] The self-term for combination is the channel kernel $h_k^{\mathsf H}(\mathbf r,\mathbf s)$ with the coefficient $\mathbf{a}_{kk}^{(l)}(\mathbf r)$, not the hidden representation from the previous layer, $\mathbf d_k^{(l)}$.
\item[4)] The update is non-local. Since the aggregation and combination steps involve integrals over $\mathcal S_{\mathrm B}$ and $\mathcal S_{\mathrm U}^k$, the updated function at any point $\mathbf s$ depends on information from the entire apertures of the BS and users.
\end{itemize}

These observations motivate a new aggregation and combination design. We interpret the beamforming function from the $l$-th iteration, $\mathbf v_k^{(l)}(\mathbf s)$, as the hidden representation in the $l$-th layer of our network, denoted by ${\mathbf d}_k^{(l)}(\mathbf s)\in\mathbb C^{C_l\times 1}$. Following this analogy, the summation terms in~\eqref{eq:iterative equation of vk2} guide the aggregation step, while the additive structure reflects the combination step. 

This model-driven approach leads to the following GNN update equation
\begin{equation} \label{Update function of ModelGNN}
\begin{split}\raisetag{2.2cm}
&{\mathbf{d}}_{k}^{\left( l+1 \right)}\left( \mathbf{s} \right)=\int_{\mathcal{S}_{\mathrm{U}}^{k}}h_{k}^{\mathsf{H}}\left( \mathbf{r},\mathbf{s} \right) {\mathbf{a}}^{(l)}_{kk}\left( \mathbf{r}\right)\mathrm{d}\mathbf{r}{\mathbf{S}}^{(l)}+\\
&\quad\quad\quad\quad\quad\quad\sum_{i=1}^K{\sum_{\substack{j=1\\ (i,j)\ne (k,k)}}^K\int_{\mathcal{S} _{\mathrm{U}}^{i}}h_{i}^{\mathsf{H}}\left( \mathbf{r},\mathbf{s} \right) {\mathbf{a}}^{(l)}_{ij}\left( \mathbf{r} \right)}\mathrm{d}\mathbf{r}{\mathbf{W}}^{(l)},
\end{split}
\end{equation}
where ${\mathbf{a}}^{(l)}_{ij}(\mathbf{r})=\int_{\mathcal{S}_{\mathrm{B}}}{h_i\left( \mathbf{r},\mathbf{s} \right) \mathbf{d}_{j}^{\left( l \right)}\left( \mathbf{s} \right)}\mathrm{d}\mathbf{s}$ is the coefficient function, and the matrices $\mathbf{\Theta}_k$ and $\mathbf{\Sigma}^k_{ij}$ from the original WMMSE algorithm in~\eqref{eq:iterative equation of vk2} are absorbed into the learnable parameter matrices ${\mathbf{S}}^{(l)}$ and ${\mathbf{W}}^{(l)}$, which are shared across all~users.

This update equation can be further refined by exploiting the graph topology. Considering that vertex $k$ only needs to aggregate information from its neighbors (i.e., the vertices connected to $k$), the double summation in~\eqref{Update function of ModelGNN} can be simplified to a single sum over the interfering users $i\neq k$~as
\begin{equation} \label{Update function of ModelGNN topology} \begin{split} &{\mathbf{d}}_{k}^{\left( l+1 \right)}\left( \mathbf{s} \right)=\int_{\mathcal{S}_{\mathrm{U}}^{k}}h_{k}^{\mathsf{H}}\left( \mathbf{r},\mathbf{s} \right) {\mathbf{a}}^{(l)}_{kk}\left( \mathbf{r}\right)\mathrm{d}\mathbf{r}{\mathbf{S}}^{(l)}+\\ &\quad\quad\quad\quad\quad\quad{\sum_{i=1,i\ne k}^K\int_{\mathcal{S} _{\mathrm{U}}^{k}}h_{k}^{\mathsf{H}}\left( \mathbf{r},\mathbf{s} \right) {\mathbf{a}}^{(l)}_{ki}\left( \mathbf{r} \right)}\mathrm{d}\mathbf{r}{\mathbf{W}}^{(l)}. \end{split} \end{equation}

This model-driven update in~\eqref{Update function of ModelGNN topology} forms the core of our proposed INR architecture, which is referred to as BeamINR.

\section{Simulation Results} \label{Simulation}
This section evaluates the proposed functional WMMSE algorithm and BeamINR against relevant numerical and learning-based baselines.

\subsection{Simulation Setup}
Unless otherwise specified, all simulations are conducted using the following setup. The number of users is $K=3$. The BS is equipped with a square CAPA of side length $L_{\mathrm{B}}^x=L_{\mathrm{B}}^y=2\,\text{m}$, while each user has a square CAPA of side length $L_{\mathrm{U}}^x=L_{\mathrm{U}}^y=0.5\,\text{m}$. The position of user $k$, $\mathbf{r}^k_o=[r_{x}^k, r_{y}^k, r_{z}^k]^{\mathsf T}$, is uniformly distributed in a region with $r_{x}^k, r_{y}^k\in [-5,5]\,\text{m}$ and $r_{z}^k\in[20,30]\,\text{m}$. Each user's orientation is randomized, with rotation angles $\omega^k_x, \omega^k_y, \omega^k_z$ drawn independently from \(\mathcal{U}(-\frac{\pi}{2}, \frac{\pi}{2})\). 
The system operates at a carrier frequency of $f=2.4\,\text{GHz}$, which corresponds to a wavelength of $\lambda=0.125\,\text{m}$, with an intrinsic impedance of $\eta=120\,\pi\Omega$. To maximize the multiplexing gain as derived in~\cite{Beamforming_Design}, the number of data streams is set to $d=\min\{d_{\mathrm{B}}, d_{\mathrm{U}}\}$, where $d_{\mathrm{B}} = \big(2\big\lceil \frac{L_{{\mathrm{B}}}^x}{\lambda} \big\rceil + 1\big)\big(2\big\lceil \frac{L_{{\mathrm{B}}}^y}{\lambda} \big\rceil + 1\big)$ and $d_{\mathrm{U}} = \big(2\big\lceil \frac{L_{{\mathrm{U}}}^x}{\lambda} \big\rceil + 1\big)\big(2\big\lceil \frac{L_{{\mathrm{U}}}^y}{\lambda} \big\rceil + 1\big)$. The maximum current budget is $C_{\max}=500\,\text{mA}^2$, and the noise variance is $\sigma_n^2=5.6\times10^{-3}\,\text{V}^2$.

BeamINR is configured with six hidden layers of dimensions $\{64, 128, 512, 512, 128, 64\}$, and the Tanh activation function is applied to each hidden layer. Following the methodology of INR training in~\cite{Implicit_Neural_CAPA}, BeamINR is trained in an unsupervised manner by minimizing the negative weighted sum of two numerical estimates of the objective in~\eqref{P1:capacity}. These two estimates are computed using different integral schemes: one with fixed Gauss--Legendre (GL) quadrature points and the other with randomized sampling points. We generate 500,000 samples for training and use a separate set of 10,000 samples for testing. The model is trained using the Adam optimizer with an initial learning rate of $10^{-3}$ and a batch size of $16$. All simulations are performed on a computer equipped with an Intel Core i9-9980XE CPU (3.00 GHz) and an NVIDIA GeForce RTX 2080 Ti GPU.

\subsection{Performance Analysis} 
We compare the proposed functional WMMSE algorithm and BeamINR with the following baselines.

\begin{itemize}
    \item \textbf{FNNINR}: An INR-based method from~\cite{Implicit_Neural_CAPA}, which uses an FNN to parameterize the beamforming function. For a fair comparison, it is trained using our derived sum-rate expression to compute the loss function.

    \item \textbf{Fourier}: The Fourier-basis method from~\cite{Wavenumber}, which represents the beamforming functions using a finite number of Fourier basis functions, transforming the problem into coefficient optimization.

    \item \textbf{SPDA}: A discretization-based method from~\cite{Pattern_Division}, which approximates a CAPA with an array of discrete~antennas.
\end{itemize}  

\begin{figure}[!t]
\centering
 \includegraphics[width=0.45\textwidth]{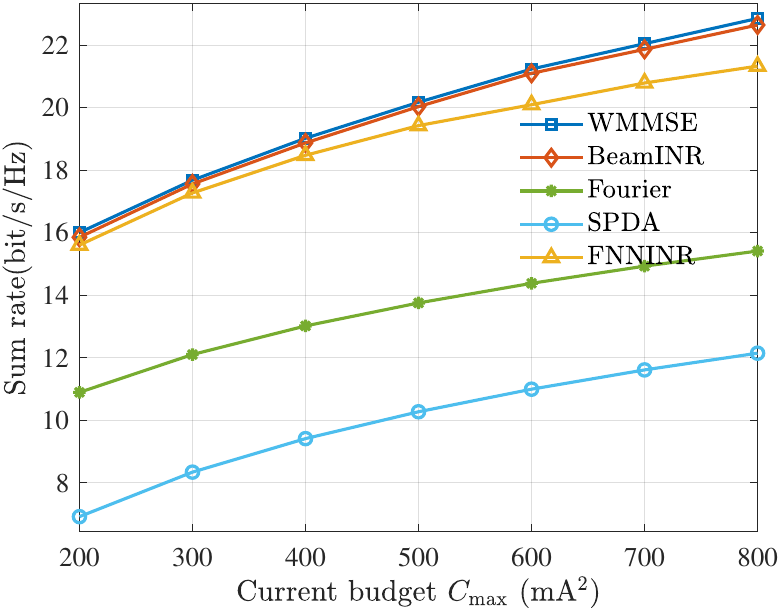}
\caption{Sum rate versus current budget.}  \label{SE_Power}
\end{figure}

%\vspace{-0.3cm}
\begin{figure}[!t]
\centering
 \includegraphics[width=0.45\textwidth]{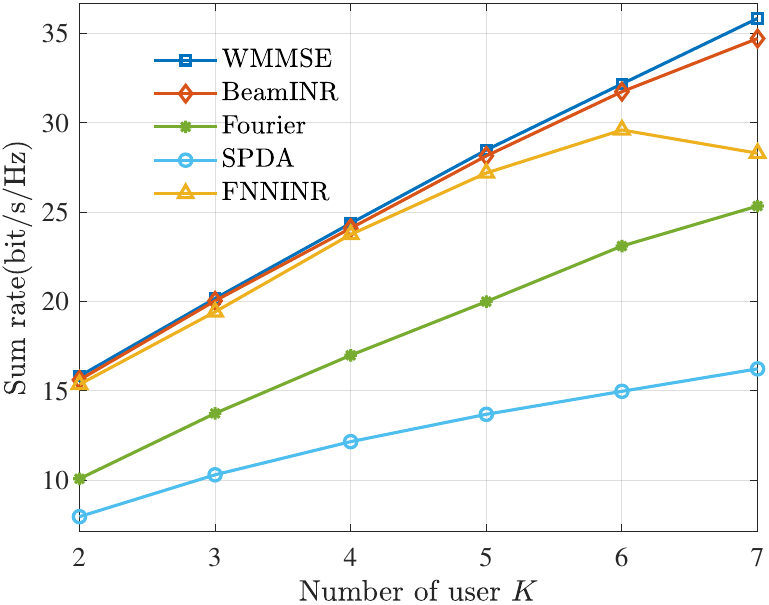}
\caption{Sum rate versus number of users.}  \label{SE_User}
\end{figure}

Fig.~\ref{SE_Power} shows the sum rates under different current budgets. As expected, the sum rate for all methods increases with the available current. Both the proposed functional WMMSE algorithm and the INR-based methods (BeamINR and FNNINR) consistently outperform the Fourier and SPDA methods because they avoid the discretization loss caused by the finite-dimensional approximations. BeamINR achieves performance nearly identical to the iterative functional WMMSE algorithm and outperforms FNNINR. This superiority stems from our model-driven design of the update equation in~\eqref{Update function of ModelGNN topology}, which incorporates both the PE property and the structure of the functional WMMSE algorithm. Fig.~\ref{SE_User} shows the sum rates under different numbers of users. The functional WMMSE algorithm and BeamINR maintain their gains over the baselines, demonstrating stronger capability for multiuser interference mitigation. In contrast, the performance of FNNINR degrades when the number of users is large. Its FNN architecture causes the number of model parameters to grow quadratically with the number of users, which makes the model increasingly difficult to train effectively.

%\vspace{-0.3cm}
\begin{figure}[!t]
\centering
 \includegraphics[width=0.45\textwidth]{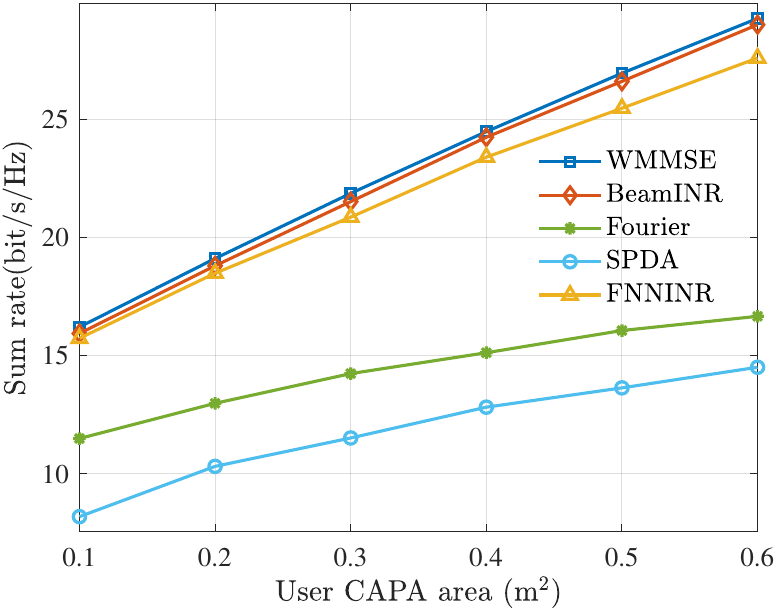}
\caption{Sum rate versus user CAPA area.}  \label{SE_Size_User}
%\vspace{-0.5cm}
\end{figure}

%\vspace{-0.3cm}
\begin{figure}[!t]
\centering
 \includegraphics[width=0.45\textwidth]{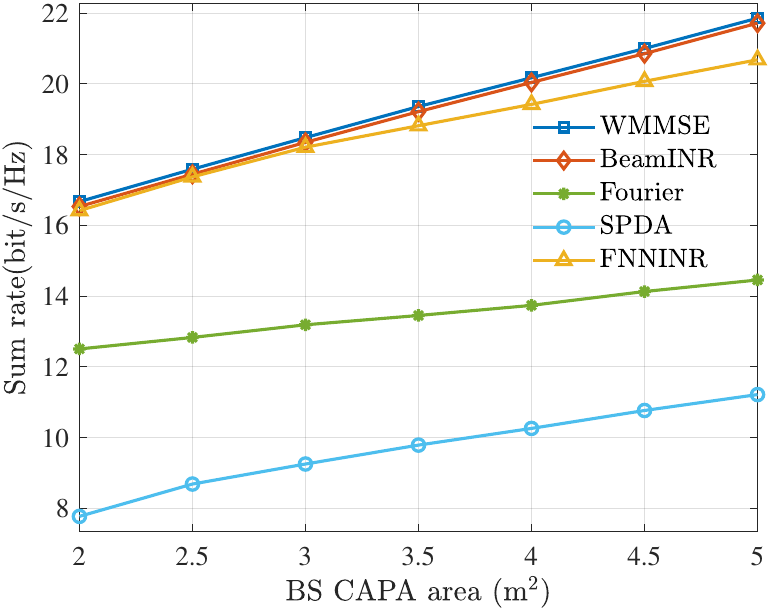}
\caption{Sum rate versus BS CAPA area.}  \label{SE_Size_BS}
%\vspace{-0.5cm}
\end{figure}

\begin{figure}[!t]
\centering
 \includegraphics[width=0.45\textwidth]{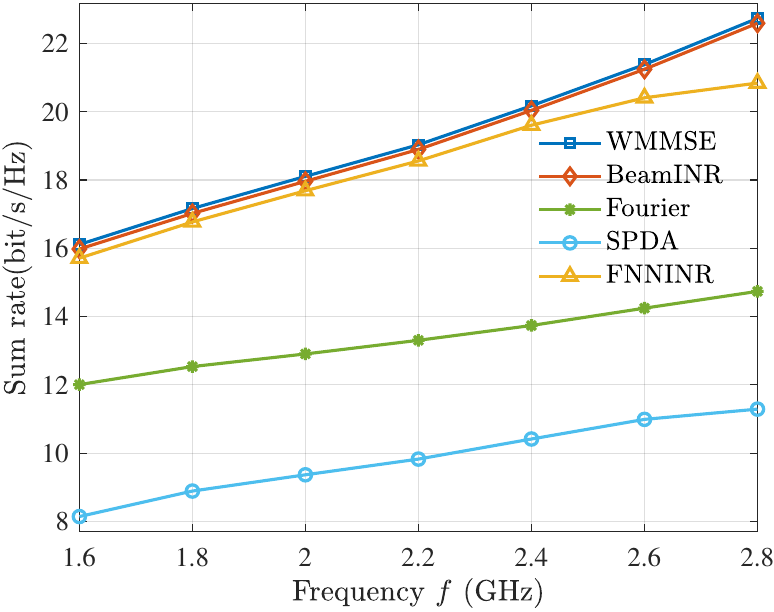}
\caption{Sum rate versus frequency.}  \label{SE_Frequency}
\end{figure}

Figs.~\ref{SE_Size_User} and~\ref{SE_Size_BS} evaluate the impact of the user and BS CAPA sizes, respectively, for square apertures ($L_{\mathrm{U}}^x=L_{\mathrm{U}}^y$ and $L_{\mathrm{B}}^x=L_{\mathrm{B}}^y$). As expected, a larger aperture provides more spatial degrees of freedom, increasing the sum rate for all methods. The performance gap between our proposed methods (functional WMMSE and BeamINR) and the baselines increases as the apertures grow. 
A similar trend is observed in Fig.~\ref{SE_Frequency}, which examines the effect of carrier frequency. The sum rate increases with frequency, consistent with the fact that smaller wavelengths provide greater spatial resolution and more available degrees of freedom~\cite{Enabling_6G_Performance}. Across the tested frequency range, BeamINR achieves performance comparable to that of the functional WMMSE algorithm and outperforms the three baselines.

Table~\ref{Learning_Performance_Num} shows the performance of the INR-based methods under different numbers of training samples. The performance metric is defined as the ratio of the sum rate achieved by each INR-based method to that of the functional WMMSE algorithm. BeamINR consistently outperforms FNNINR across different numbers of training samples. Notably, BeamINR reaches over $95\%$ of the functional WMMSE's performance with only $10$ training samples, whereas FNNINR remains below \(93\%\) even with \(5000\) training samples. This superior sample efficiency is a direct benefit of the designed GNN update equation.

\begin{table}[h]
% \vspace{-0.2cm}
\captionsetup{font=small}
\renewcommand{\arraystretch}{1.3}
\centering
\caption{Learning Performance Versus Number of Training Samples}
\label{Learning_Performance_Num}
\setlength{\tabcolsep}{3pt}
\begin{tabular}{@{}c|*7{c}@{}}
\hline
Number of samples & 5  & 10  & 50  & 100  & 500  & 1000  & 5000 \\
\hline
BeamINR (\%) & 88.19 & 95.23 & 96.92 & 97.38 & 98.62 & 99.22 & 99.30 \\
FNNINR (\%)  & 6.25 & 7.67 & 17.85 & 57.04 & 90.67 & 92.03 & 92.65 \\
\hline
\end{tabular}
% \vspace{-0.2cm}
\end{table}

\subsection{Inference and Training Complexity}
Table~\ref{Complexity} compares the inference time and training complexity of the considered methods. For the INR-based methods, training complexity is measured by the number of samples, training time, and model parameters (space cost) required to reach $95\%$ of the sum rate achieved by the functional WMMSE~algorithm.

The results show that BeamINR and FNNINR offer drastically reduced inference times compared to the numerical baselines. Compared with FNNINR, BeamINR requires significantly fewer training samples, shorter training time, and lower space cost. BeamINR exhibits a slightly higher latency than FNNINR. This is an expected tradeoff for its update equation, which introduces additional computations not present in an~FNN.

\begin{table}[htbp]	
% \vspace{-0.1cm}
\captionsetup{font=small}
\renewcommand{\arraystretch}{1.2}
\centering 
% \small
\caption{Inference Time and Training Complexity}	
\begin{tabular}{c|c|c|c|c}
\hline
\multirow{2}{*}{Name} & \multirow{2}{*}{Inference time} & \multicolumn{3}{c}{Training Complexity} \\
\cline{3-5}
                     &                                  & Sample  & Time      & Space \\ \hline             
\textbf{BeamINR}      & 0.052 s                         & 10     & 0.81 h   & 3.68 M    \\ 
\cline{1-5}
\textbf{FNNINR}      & 0.032 s                         & 20 K   & 41.9 h   & 18.11 M    \\
\cline{1-5}
\textbf{WMMSE}     & 1.378 s                           & -- & --  & --  \\ 
\cline{1-5}
\textbf{Fourier}     & 10.88 s                         & -- & --  & --  \\ 
\cline{1-5}
\textbf{SPDA}     & 9.590 s                            & -- & --  & --  \\ 
\hline
\end{tabular}
\begin{minipage}{0.95\linewidth}
\quad\footnotesize \textit{Note:} ``K'' and ``M'' represent thousand and million, respectively.
\end{minipage}
\label{Complexity}
% \vspace{-0.1cm}
\end{table}

\subsection{Ablation Study}
To validate and quantify the contribution of the two main designs in BeamINR, we conduct an ablation study. Specifically, we isolate the WMMSE-inspired design for the update equation and the topology-aware simplification, and create two ablated variants as follows.

\begin{itemize}
    \item \textbf{ConINR}: This variant uses the conventional GNN update equation in~\eqref{Update function of GNN}. It therefore relies solely on the PE property without guidance from the structure of the functional WMMSE algorithm.

    \item \textbf{VarINR}: This variant uses the update equation in~\eqref{Update function of ModelGNN} but omits the topology-aware simplification.
\end{itemize}
The performance metric in this subsection is the ratio of the sum rate achieved by each INR to that achieved by the proposed functional WMMSE algorithm.

\begin{table}[h]
% \vspace{-0.2cm}
\captionsetup{font=small}
\centering
\caption{Learning Performance Versus Current Budgets}
\renewcommand{\arraystretch}{1.2}
\label{Learning_Performance}
\setlength{\tabcolsep}{3pt}
\begin{tabular}{@{}c|*7{c}@{}}
\hline
$C_{\max}$ (mA$^2$) & 200  & 300  & 400  & 500  & 600  & 700  & 800 \\
\hline
BeamINR (\%) & 99.97 & 99.32 & 98.32 & 98.30 & 97.91 & 97.74 & 97.48 \\
VarINR (\%)  & 98.78 & 98.19 & 97.53 & 96.72 & 95.75 & 94.52 & 92.31 \\
ConINR (\%)  & 97.33 & 97.06 & 96.94 & 96.06 & 94.22 & 93.10 & 92.08 \\
\hline
\end{tabular}
% \vspace{-0.2cm}
\end{table}

Table~\ref{Learning_Performance} shows the learning performance of the INR-based methods under various current budgets. Both BeamINR and VarINR outperform ConINR, confirming the substantial performance gain from incorporating the functional WMMSE iterative structure into the GNN. The advantage of BeamINR over VarINR is evident under high current budgets, demonstrating the benefit from further considering the topology-aware~simplification.

\begin{table}[h]
% \vspace{-0.2cm}
\captionsetup{font=small}
\centering
\caption{User Generalization}
\renewcommand{\arraystretch}{1.2}
\label{User_Generalizability}
\setlength{\tabcolsep}{3pt}
\begin{tabular}{@{}c|*7{c}@{}}
\hline
$K$    & 2   & 3     & 4   & 5   & 6   & 7   & 8    \\
\hline
BeamINR (\%)  & 71.80  & 86.29 & 90.71 & 98.85 & 91.71 & 90.31 & 86.51 \\
VarINR (\%)   & 37.61  & 42.72 & 69.14 & 96.71 & 75.25 & 57.83 & 56.23 \\
ConINR (\%)   & 55.86  & 81.68 & 85.62 & 95.13 & 84.08 & 80.42 & 76.26\\
\hline
\end{tabular}
% \vspace{-0.2cm}
\end{table}

Table~\ref{User_Generalizability} evaluates the generalization performance to the number of users. All INRs are trained with $K=5$ and tested for $K=2, \dots, 8$. BeamINR achieves the highest performance for all tested user numbers. For $K=5$, BeamINR and VarINR achieve comparable performance. However, when generalizing to unseen numbers of users, VarINR suffers a severe performance degradation, suggesting that merely incorporating the structure of the functional WMMSE algorithm is insufficient for generalization. In contrast, ConINR exhibits more stable generalization performance, which implies that combining the WMMSE-inspired update with the topology-aware simplification is crucial to achieve both high learning and generalization performance across different numbers of~users.

\begin{table}[h]
% \vspace{-0.15cm}
\centering
\captionsetup{font=small}
\caption{CAPA Size Generalization}
\renewcommand{\arraystretch}{1.2}
\label{Size_Generalizability}
\begin{tabular}{c|ccccc}\hline
$A_{\mathrm{B}}$ (m$^2$) & 2 & 3 & 4 & 5 & 6 \\
\hline
BeamINR (\%) & 96.62  & 96.92 & 98.30 & 97.35 & 96.56 \\
VarINR (\%)  & 95.34  & 95.93 & 96.72 & 96.84 & 95.54 \\
ConINR (\%)  & 94.21  & 94.67 & 96.06 & 94.83 & 94.01 \\
\hline
$A_{\mathrm{U}}$ (m$^2$) & 0.3 & 0.4 & 0.5 & 0.6 & 0.7 \\
\hline
BeamINR (\%) & 93.51 & 94.87 & 97.39 & 96.32 & 95.02 \\
VarINR (\%)  & 91.07 & 93.36 & 96.52 & 95.07 & 93.60 \\
ConINR (\%)  & 86.35 & 91.07 & 95.83 & 94.32 & 92.93 \\
\hline
\end{tabular}
% \vspace{-0.15cm}
\end{table}

Table~\ref{Size_Generalizability} evaluates the generalization performance with respect to CAPA size. 
For BS-side generalization, all models are trained with $A_{\mathrm{B}}=4\,\mathrm{m}^2$ and tested over $A_{\mathrm{B}}\in[2,6]\,\mathrm{m}^2$ with $L_{\mathrm{B}}^x=L_{\mathrm{B}}^y$. For user-side generalization, the models are trained with $A_{\mathrm{U}}=0.5\,\mathrm{m}^2$ and tested over $A_{\mathrm{U}}\in[0.3,0.7]\,\mathrm{m}^2$ with $L_{\mathrm{U}}^x=L_{\mathrm{U}}^y$. It is shown that all three methods generalize well to BS CAPA sizes, where BeamINR consistently achieves the highest performance. For user CAPA generalization, BeamINR also outperforms the two ablated variants, especially when the user aperture is smaller than the training~value. 

\begin{table}[h]
% \vspace{-0.2cm}
\captionsetup{font=small}
\centering
\caption{Frequency Generalization}
\renewcommand{\arraystretch}{1.2}
\label{Frequency_Generalizability}
\setlength{\tabcolsep}{3pt}
\begin{tabular}{@{}c|*7{c}@{}}
\hline
$f$ (GHz)    & 1.8   & 2     & 2.2   & 2.4   & 2.6   & 2.8   & 3    \\
\hline
BeamINR (\%) & 90.84 & 92.20 & 95.82 & 98.30 & 95.64 & 92.31 & 90.52 \\
VarINR (\%)  & 88.92 & 90.26 & 93.06 & 96.72 & 94.59 & 90.01 & 89.35 \\
ConINR (\%)  & 81.36 & 88.56 & 92.73 & 96.06 & 93.02 & 86.96 & 80.23 \\
\hline
\end{tabular}
% \vspace{-0.2cm}
\end{table}

Table~\ref{Frequency_Generalizability} evaluates the generalization performance with respect to carrier frequency. All models are trained at $f=2.4\,\text{GHz}$ and tested over $f\in[1.8,3]\,\text{GHz}$. ConINR's performance degrades quickly as the frequency deviates from its training point, whereas both VarINR and BeamINR exhibit better generalization. This improvement follows from their model-driven design. Their WMMSE-inspired update equations in~\eqref{Update function of ModelGNN} and~\eqref{Update function of ModelGNN topology} involve the channel kernels. As shown in~\eqref{Green's function}, the channel kernel is a function of wavelength $\lambda$ and hence of the carrier frequency, which enables VarINR and BeamINR to adapt to frequency variations. 

\begin{table}[htbp]	
% \vspace{-0.1cm}
\captionsetup{font=small}
\centering 
\caption{Inference Time and Training Complexity of INR Variants}	
\renewcommand{\arraystretch}{1.2}
\begin{tabular}{c|c|c|c|c}
\hline
\multirow{2}{*}{Name} & \multirow{2}{*}{Inference time} & \multicolumn{3}{c}{Training Complexity} \\
\cline{3-5}
                     &                                  & Sample  & Time      & Space \\ \hline             
\textbf{BeamINR} & 0.052 s & 10 & 0.81 h & 3.68 M \\ 
\cline{1-5}
\textbf{VarINR}  & 0.082 s & 1 K & 32.1 h & 4.72 M \\
\cline{1-5}
\textbf{ConINR}  & 0.033 s & 5 K & 3.68 h & 8.63 M \\ 
\hline
\end{tabular}
\begin{minipage}{0.95\linewidth}
\quad\footnotesize \textit{Note:} ``K'' and ``M'' represent thousand and million, respectively.
\end{minipage}
\label{Complexity_Ablation}
% \vspace{-0.1cm}
\end{table}

Table~\ref{Complexity_Ablation} compares the inference latency and training complexity of the three methods, where training complexity is measured by the number of samples, training time, and space cost required to reach $95\%$ of the sum rate achieved by the functional WMMSE algorithm. BeamINR requires the fewest training samples, shortest training time, and lowest space cost. Compared with VarINR, BeamINR reduces both training and inference complexity by eliminating redundant aggregation terms through the topology-aware simplification. Although ConINR offers the lowest inference latency, it requires more training samples and higher space cost.

\section{Conclusions}
This paper investigated beamforming optimization in multiuser multi-CAPA systems, where both the BS and users are equipped with CAPAs. We first derived a closed-form expression for the achievable sum rate and then proposed a functional WMMSE algorithm for sum-rate maximization. Building on this algorithm, we further designed BeamINR to learn the beamforming function. BeamINR exploits the PE property of the optimal beamforming policy, incorporates a functional WMMSE-inspired update equation, and further employs a topology-aware simplification to improve the learning efficiency.
Simulation results demonstrated that the proposed functional WMMSE algorithm outperforms existing numerical optimization baselines. BeamINR achieves sum-rate performance comparable to that of the functional WMMSE algorithm while substantially reducing inference latency. Compared with INR-based baselines, BeamINR requires fewer training samples, shorter training time and lower memory cost, and exhibits stronger generalization across different numbers of users, CAPA sizes, and carrier frequencies.

\appendices
\numberwithin{equation}{section}  
\section{Proof of Proposition~\ref{proposition_capacity}} \label{appendix_capacity}
From~\eqref{receive yk}, the achievable sum rate of user $k$ is the mutual information between the data vector and the received signal,~i.e.,
\begin{equation}
    R_k=H(y_k(\cdot))-H(y_k(\cdot)\!\!\mid\!\! \mathbf{x}_k),\label{eq:capacity}
\end{equation}
where $y_k(\cdot)$ denotes the received-signal function on $\mathcal{S}_{\mathrm{U}}^k$, and $H(\cdot)$ and $H(\cdot\!\!\mid\!\! \cdot)$ denote the differential entropy and conditional differential entropy, respectively.

In~\eqref{eq:capacity}, both $y_k(\mathbf r)$ and $n_k(\mathbf r)$ are Gaussian processes. To evaluate their differential entropies, we employ the Karhunen--Loève expansion (KLE), which represents a Gaussian process in an orthonormal basis and yields a sequence of independent Gaussian random variables. Thus, as the number of basis functions tends to infinity, the entropy of the process can be characterized through the entropy of these Gaussian coefficients. Let $\bm{\upalpha}_k(\mathbf r) = [\upalpha^k_1(\mathbf r), \ldots, \upalpha^k_{N_r}(\mathbf r)] \in \mathbb{C}^{1\times N_r}$ denote an orthonormal basis on $\mathcal{S}_{\mathrm{U}}^k$ with $N_r \to \infty$, satisfying
\begin{equation}\label{eq:complete set}
   \int_{\mathcal{S}_{\mathrm{U}}^k} \bm{\upalpha}_k^{\mathsf H}(\mathbf r)\bm{\upalpha}_k(\mathbf r)\mathrm d\mathbf r = \mathbf I_{N_r}.
\end{equation}
Under this basis, the noise process $n_k(\mathbf r)$ admits the expansion
\begin{equation}\label{eq:n expansion}
    n_k(\mathbf r) = \bm{\upalpha}_k(\mathbf r)\mathbf{n}_k,
\end{equation}
where $\mathbf{n}_k \in \mathbb{C}^{N_r\times 1}$ and $\mathbf{n}_{k} \sim \mathcal{CN}(\mathbf 0,\sigma_n^2\mathbf I_{N_r})$. Since $\bm{\upalpha}_k(\mathbf r)$ is complete on $\mathcal{S}^k_{\mathrm{U}}$, each $\mathbf{a}_{ki}(\mathbf r)$ admits the expansion 
\begin{equation}\label{eq:a expansion}
    \mathbf{a}_{ki}(\mathbf r)=\bm{\upalpha}_k(\mathbf r)\mathbf A_{ki},
\end{equation}
where $\mathbf A_{ki}\in\mathbb{C}^{N_r\times d}$ is obtained by the following projection
\begin{equation} \label{eq:a projection coefficient}
    \mathbf A_{ki}=\int_{\mathcal{S}_{\mathrm{U}}^k} \bm{\upalpha}_k^{\mathsf H}(\mathbf r)\mathbf{a}_{ki}(\mathbf r)\mathrm d\mathbf r.
\end{equation}
Substituting~\eqref{eq:n expansion} and~\eqref{eq:a expansion} into~\eqref{receive yk} yields
\begin{equation}\label{eq:transformation}
\begin{split}
    y_k(\mathbf{r})&=\sum\limits_{i=1}^K\bm{\upalpha}_k(\mathbf r)\mathbf{A}_{ki}\mathbf{x}_i+\bm{\upalpha}_k(\mathbf r)\mathbf{n}_k,\\
&=\bm{\upalpha}_k(\mathbf r)\Big(\sum\limits_{i=1}^K\mathbf{A}_{ki}\mathbf{x}_i+\mathbf{n}_k\Big)\triangleq\bm{\upalpha}_k(\mathbf r)\mathbf{y}_k.
\end{split}
\end{equation}

For a linear transformation $\tilde{\mathbf y}=\mathbf A\mathbf y$, the differential entropy and conditional differential entropy satisfy $H(\tilde{\mathbf y})=H(\mathbf y)+\log|\det(\mathbf A)|$ and $H(\tilde{\mathbf y}\mid \mathbf x)=H(\mathbf y\mid \mathbf x)+\log|\det(\mathbf A)|$, respectively~\cite{Elements_of_Information_Theory}. With~\eqref{eq:transformation}, applying these identities to~\eqref{eq:capacity}~yields
\begin{equation}\label{eq:capacity2}
\begin{split}%\raisetag{1.3cm}
    R_k&=H(\mathbf{y}_k)+\log|\det(\bm{\upalpha}_k(\cdot))|\\
&\quad\quad\quad-\big(H(\mathbf{y}_k\mid \mathbf{x}_k)+\log|\det(\bm{\upalpha}_k(\cdot))|\big)\\
&=H(\mathbf{y}_k)-H(\mathbf{y}_k\mid \mathbf{x}_k).
\end{split}
\end{equation}
Since $\mathbf{y}_k$ is a Gaussian random vector, its entropy and conditional entropy are given by
\begin{subequations}\label{eq:entropy and conditional entropy of gaussian_vector}
\begin{align}
&\!\!\!    H(\mathbf{y}_{k})\!=\!\!\!\!\underset{N_r\to\infty}{\lim}\!\!\!\log\!\det\!\big(\!(\pi e)^{N_r}\!(\sum\limits_{i=1}^K\mathbf{A}_{ki}\mathbf{A}_{ki}^{\mathsf H}\!+\!\sigma_n^2\mathbf{I}_{N_r}) \!\big), \label{eq:entropy of gaussian_vector}\\
&\!\!\!    H(\mathbf{y}_{k}\!\!\mid\!\! \mathbf{x}_k)\!\!=\!\!\!\!\!\underset{N_r\to\infty}{\lim}\!\!\!\log\!\det\!\!\big(\!(\pi e\!)^{N_r}\!(\!   \sum\limits_{\substack{i=1\\ i\ne k}}^{K}\!\mathbf{A}_{ki}\mathbf{A}_{ki}^{\mathsf H}\!+\!\!\sigma_n^2\mathbf{I}_{N_r}\!)\!\big).\label{eq:conditional entropy of gaussian_vector}
\end{align}
\end{subequations}
Substituting~\eqref{eq:entropy and conditional entropy of gaussian_vector} into~\eqref{eq:capacity2} yields
\begin{equation}\label{eq:capacity3}
\begin{split}\raisetag{0.6cm}
    R_k&=\!\!\underset{N_r\to\infty}{\lim}\!\!\!\log\!\det\!\Big(\mathbf{I}_{d}\!\!+\mathbf{A}^{\mathsf H}_{kk}\big(   \sum\limits_{\substack{i=1\\ i\ne k}}^{K}\mathbf{A}_{ki}\mathbf{A}_{ki}^{\mathsf H}+\sigma_n^2\mathbf{I}_{N_r})^{-1}\mathbf{A}_{kk}\big)\Big)\\
    &\triangleq\!\!\underset{N_r\to\infty}{\lim}\!\!\!\log\!\det\!\big(\mathbf{I}_{d}\!\!+\tilde{\mathbf{Q}}_k\big).
\end{split}
\end{equation}

Next, we convert \eqref{eq:capacity3} from the coefficient domain to the functional domain using~\eqref{eq:complete set}$\sim$\eqref{eq:a expansion}. Defining $\tilde{\mathbf{A}}_{k}=[\mathbf{A}_{ki}]_{i\neq k}\in\mathbb{C}^{N_r\times (K-1)d}$ and applying the Woodbury matrix identity, we obtain
\begin{equation}\label{eq:Q matrix}
\begin{split}\raisetag{1.4cm}
&   \tilde{\mathbf{Q}}_k=\mathbf{A}^{\mathsf{H}}_{kk}\big(\tilde{\mathbf{A}}_{k}\tilde{\mathbf{A}}^{\mathsf H}_{k}+\sigma_n^2\mathbf{I}_{N_r})^{-1}\mathbf{A}_{kk}\\
&   =\frac{1}{\sigma_n^2}\mathbf{A}^{\mathsf{H}}_{kk}\mathbf{A}_{kk}\!-\!\!\frac{1}{\sigma_n^4}\mathbf{A}^{\mathsf{H}}_{kk}\tilde{\mathbf{A}}_{k}(\mathbf{I}_{(K-1)d}\!+\!\frac{1}{\sigma_n^2}\tilde{\mathbf{A}}^{\mathsf H}_{k}\tilde{\mathbf{A}}_{k})^{-1}\!\!\tilde{\mathbf{A}}^{\mathsf H}_{k}\mathbf{A}_{kk}.
\end{split}
\end{equation}
Using~\eqref{eq:complete set}$\sim$\eqref{eq:a expansion}, each matrix product in~\eqref{eq:Q matrix} can be written as a functional integral. For example,
\begin{equation}
    \begin{split}\label{eq:matrix into integral akk}
\mathbf{A}^{\mathsf{H}}_{kk}\mathbf{A}_{kk}&= \int_{\mathcal{S}^k_{\mathrm{U}}}\underset{\mathbf{a}^{\mathsf H}_{kk}\left( \mathbf{r}\right)}{\underbrace{\mathbf{A}^{\mathsf{H}}_{kk} {\bm{\upalpha }^{\mathsf{H}}_k\left( \mathbf{r} \right)} }}\underset{\mathbf{a}_{kk}\left( \mathbf{r} \right)}{\underbrace{\bm{\upalpha}_k\left( \mathbf{r} \right)\mathbf{A}_{kk}}}\mathrm{d}\mathbf{r}\\
&   = \int_{\mathcal{S}^k_{\mathrm{U}}}{\mathbf{a}^{\mathsf{H}}_{kk}\left( \mathbf{r}\right) \mathbf{a}_{kk}\left( \mathbf{r} \right)}\mathrm{d}\mathbf{r}.
    \end{split}
\end{equation}
Applying the same conversion to the remaining terms in~\eqref{eq:Q matrix}~yields
\begin{equation}\label{eq:Q matrix1}
\begin{split}\raisetag{1.35cm}
\tilde{\mathbf{Q}}_k=\iint_{\mathcal{S}^k_{\mathrm{U}}}{\mathbf{a}^{\mathsf{H}}_{kk}\left( \mathbf{r}_1\right)\mathrm{E}_k\left( \mathbf{r}_1,\mathbf{r}_2 \right)}\mathbf{a}_{kk}\left( \mathbf{r}_2 \right)\mathrm{d}\mathbf{r}_1\mathrm{d}\mathbf{r}_2,
\end{split}
\end{equation}
where $\mathrm{E}_k\left( \mathbf{r}_1,\mathbf{r}_2 \right)\!=\!\mathrm{Y}\!\!\left( \mathbf{r}_1,\mathbf{r}_2 \right)\!-\!\bm{\uppsi}(\mathbf{r}_1)\left( \mathbf{I}_{(K-1)d}\!+\!\!\mathbf{G} \right)^{\!-1}\!\! \bm{\upphi}(\mathbf{r}_2)$, $\mathbf{G}=\iint_{\mathcal{S}_{\mathrm{U}}}\tilde{\mathbf{a}}^{\mathsf{H}}_k(\mathbf{r}_1)\mathrm{Y}\left( \mathbf{r}_1,\mathbf{r}_2 \right)\tilde{\mathbf{a}}_k(\mathbf{r}_2)\mathrm{d}\mathbf{r}_1\mathrm{d}\mathbf{r}_2$, and the auxiliary functions are defined as $\mathrm{Y}\left( \mathbf{r}_1,\mathbf{r}_2 \right)=\frac{1}{\sigma_n^2}\delta(\mathbf{r}_1-\mathbf{r}_2)$, $\tilde{\mathbf{a}}_k(\mathbf{r})=[\mathbf{a}_{ki}(\mathbf{r})]_{i\neq k}\in\mathbb{C}^{1\times (K-1)d}$, and 
\begin{subequations}\raisetag{5cm}
    \begin{align}
\bm{\uppsi}(\mathbf{r}_1)&=   \int_{\mathcal{S}_{\mathrm{U}}}\mathrm{Y}\left( \mathbf{r}_1,\mathbf{r}_2 \right)\tilde{\mathbf{a}}_k(\mathbf{r}_2)\mathrm{d}\mathbf{r}_2,\\
\bm{\upphi}(\mathbf{r}_2)&=   \int_{\mathcal{S}_{\mathrm{U}}}\tilde{\mathbf{a}}^{\mathsf{H}}_k(\mathbf{r}_1)\mathrm{Y}\left( \mathbf{r}_1,\mathbf{r}_2 \right)\mathrm{d}\mathbf{r}_1.
    \end{align}
\end{subequations}

\begin{lemma}[Functional Woodbury Identity]\label{lem:woodbury}
Let $\mathrm{J}\left( \mathbf{r}_1,\mathbf{r}_2 \right)$ be a continuous invertible kernel, and let $\mathbf{a}(\mathbf{r}_1)\in \mathbb{C} ^{1\times d}$ and $\mathbf{b}(\mathbf{r}_2)\in \mathbb{C} ^{d\times 1}$ for $\mathbf{r}_1,\mathbf{r}_2\in \mathcal{S}$. 
If the term $\mathrm{J}\left( \mathbf{r}_1,\mathbf{r}_2 \right) +\mathbf{a}(\mathbf{r}_1)\mathbf{b}\left( \mathbf{r}_2 \right)$ is invertible,~then
\begin{equation} \label{eq:lemma1 inverse of continuous kernel}
\begin{aligned}
&\left( \mathrm{J}\left( \mathbf{r}_1,\mathbf{r}_2 \right) +\mathbf{a}(\mathbf{r}_1)\mathbf{b}\left( \mathbf{r}_2\right) \right) ^{-1}
\\
&\quad\quad\quad= \mathrm{J}^{-1}\left( \mathbf{r}_1,\mathbf{r}_2 \right) -\bm{\uppsi}(\mathbf{r}_1)\left( \mathbf{I}_d+\mathbf{G} \right)^{-1} \bm{\upphi}(\mathbf{r}_2),
\end{aligned}
\end{equation}
where $\mathbf{G}=  \iint_{\mathcal{S}}{\mathbf{b}\left( \mathbf{r}_1 \right)  \mathrm{J}^{-1}\!\left( \mathbf{r}_1,\mathbf{r}_2 \right)\mathbf{a}(\mathbf{r}_2)}\mathrm{d}\mathbf{r}_1\mathrm{d}\mathbf{r}_2$ and
\begin{equation}\label{eq:lemma1 auxiliary terms}
\begin{aligned}
&\bm{\uppsi}(\mathbf{r}_1)=\int_{\mathcal{S}}{\mathrm{J}^{-1}\left( \mathbf{r}_1,\mathbf{r}_2 \right) \mathbf{a}(\mathbf{r}_2)\mathrm{d}\mathbf{r}_2},
\\
&\bm{\upphi}(\mathbf{r}_2)=     \int_{\mathcal{S}}{\mathbf{b}\left( \mathbf{r}_1 \right)  \mathrm{J}^{-1}\!\left( \mathbf{r}_1,\mathbf{r}_2 \right) \mathrm{d}\mathbf{r}_1}.
\end{aligned}
\end{equation}
\end{lemma}
\begin{IEEEproof}
See Appendix~\ref{App:proof of Lemma}.
\end{IEEEproof}

Since $\mathrm{E}_k\left( \mathbf{r}_1,\mathbf{r}_2 \right)$ in~\eqref{eq:Q matrix1} has a structure similar to the right-hand side of~\eqref{eq:lemma1 inverse of continuous kernel}, applying Lemma~\ref{lem:woodbury} yields
\begin{equation}\label{eq:functional Ek}
    \begin{split}\raisetag{1.2cm}
       \!\!\!    \mathrm{E}_k\left( \mathbf{r}_1,\mathbf{r}_2 \right)&   =\big(\mathrm{Y}^{-1}\left( \mathbf{r}_1,\mathbf{r}_2 \right)+\tilde{\mathbf{a}}_k(\mathbf{r}_1)\tilde{\mathbf{a}}^{\mathsf{H}}_k(\mathbf{r}_2)\big)^{-1}\\
&   =\big({\sigma_n^{2}}\delta(\mathbf{r}_1-\mathbf{r}_2)\!+\! 
    \sum\limits_{\substack{i=1\\ i\ne k}}^K\! \mathbf{a}_{ki}(\mathbf{r}_1)\,\mathbf{a}_{ki}^{\mathsf{H}}(\mathbf{r}_2) 
\big)^{-1}.
    \end{split}
\end{equation}
Substituting~\eqref{eq:functional Ek} into~\eqref{eq:Q matrix1} yields $\tilde{\mathbf{Q}}_k=\mathbf{Q}_k$, where $\mathbf Q_k$ is defined in~\eqref{eq:theorem1 Q}. Substituting this result into~\eqref{eq:capacity3} and summing over $k$ completes the proof.

\section{Proof of Lemma~\ref{lem:woodbury}} \label{App:proof of Lemma}
By Definition~\ref{Definition:Inverse of continuous kernel} in Sec.~\ref{Sec:Problem Formulation}, it suffices to verify that the following term equals $\delta(\mathbf{r}_1-\mathbf{r}_2)$.
\begin{equation} \label{eq: proof of lemma1 eq1}
\begin{split}\raisetag{20.0ex}
&\int_{\mathcal{S}}\big(\big(\mathrm{J}^{-1}\left( \mathbf{r}_1,\mathbf{r} \right) -\bm{\uppsi}(\mathbf{r}_1)\left( \mathbf{I}_d+\mathbf{G} \right)^{-1}\bm{\upphi}(\mathbf{r})\big)\big.\cdot\\
&\quad\quad\big.\left(\mathrm{J}\left( \mathbf{r},\mathbf{r}_2 \right) +\mathbf{a}(\mathbf{r})\mathbf{b}\left( \mathbf{r}_2 \right) \right)\big)\mathrm{d}\mathbf{r}\\
&   =\delta \left( \mathbf{r}_1\!-\!\mathbf{r}_2 \right)\!-\bm{\uppsi}(\mathbf{r}_1)\left( \mathbf{I}_d+\mathbf{G} \right)^{-1}\!\!\!\underset{\mathbf{G}}{\underbrace{\int_{\mathcal{S}}{\bm{\upphi}(\mathbf{r})\mathbf{a}(\mathbf{r})}\mathrm{d}\mathbf{r}}}\cdot\mathbf{b}\left( \mathbf{r}_2 \right)-\\
&\bm{\uppsi}(\mathbf{r}_1)\!\left( \mathbf{I}_d\!+\!\mathbf{G} \right)^{\!-1}\!\!\!\!\underset{\mathbf{b}\left( \mathbf{r}_2 \right)}{\underbrace{\int_{\mathcal{S}}\!{\!\bm{\upphi}(\mathbf{r})\mathrm{J}\left(\mathbf{r},\!\mathbf{r}_2\right)}\mathrm{d}\mathbf{r}}}\!+\!\!\underset{\bm{\uppsi}(\mathbf{r}_1)}{\underbrace{\int_{\mathcal{S}}\!{\mathrm{J}^{\!-1}\!\!\left(\mathbf{r}_1,\!\mathbf{r}\right) \mathbf{a}(\mathbf{r})}\mathrm{d}\mathbf{r}}}\cdot\mathbf{b}\!\left(\mathbf{r}_2\right).
\end{split}
\end{equation}

Using the definitions in~\eqref{eq:lemma1 auxiliary terms}, the right-hand side of~\eqref{eq: proof of lemma1 eq1} simplifies to
\begin{equation}
\begin{split}\raisetag{3.0ex}
&\eqref{eq: proof of lemma1 eq1}=\delta \left( \mathbf{r}_1-\mathbf{r}_2 \right) -\bm{\uppsi}(\mathbf{r}_1)\left( \mathbf{I}_d+\mathbf{G} \right)^{-1}\mathbf{G}\mathbf{b}\left( \mathbf{r}_2 \right)\\
&\quad\quad\quad -\bm{\uppsi}(\mathbf{r}_1)\left( \mathbf{I}_d+\mathbf{G} \right)^{-1}\mathbf{b}\left( \mathbf{r}_2 \right)+\bm{\uppsi}(\mathbf{r}_1)\mathbf{b}\left( \mathbf{r}_2 \right)\\
&=\delta \!\left(\mathbf{r}_1\!-\!\mathbf{r}_2 \right)\!-\!\bm{\uppsi}(\mathbf{r}_1)\! \big( \!\left( \mathbf{I}_d+\mathbf{G}\right)^{\!-1}\!\mathbf{G}\!+\!\left( \mathbf{I}_d+\mathbf{G} \right)^{\!-1}\!\!-\!\mathbf{I}_d\big)\mathbf{b}\left( \mathbf{r}_2\right)\\
&=\delta \left( \mathbf{r}_1-\mathbf{r}_2 \right),
\end{split}
\end{equation}
where we used the identity $\left( \mathbf{I}_d+\mathbf{G}\right)^{-1}\mathbf{G}+\left( \mathbf{I}_d+\mathbf{G} \right)^{-1}=\mathbf{I}_d$. Hence, by Definition~\ref{Definition:Inverse of continuous kernel}, the proof is complete.

\section{Proof of Proposition~\ref{proposition:MSE}} \label{appendix:MSE}
From the first-order optimality condition of problem~\eqref{P2:MMSE} with respect to $\mathbf{u}_{k}\left( \mathbf{r} \right)$, the optimal combining function is
\begin{equation} \label{eq:optimal u}
 \mathbf{u}_{k}^{\text{opt}}\left( \mathbf{r} \right) =    \int_{\mathcal{S}^k _{\mathrm{U}}}{\mathrm{J}_{k}^{-1}\left( \mathbf{r},\mathbf{r}_1 \right) \mathbf{a}_{kk}(\mathbf{r}_1)\mathrm{d}\mathbf{r}_1},
\end{equation}
where $\mathrm{J}_{k}(\mathbf{r},\mathbf{r}_1)\!\triangleq\! \sum\limits_{j=1}^K\! \mathbf{a}_{kj}(\mathbf{r})\,\mathbf{a}_{kj}^{\mathsf{H}}(\mathbf{r}_1) 
\!+\!{\sigma_n^{2}}\,\delta(\mathbf{r}\!-\!\mathbf{r}_1)$ and $\mathrm{J}^{-1}_{k}(\mathbf{r},\mathbf{r}_1)$ denotes its inverse as defined in Definition~\ref{Definition:Inverse of continuous kernel}. Likewise, the first-order optimality condition with respect to $\mathbf{W}_{k}$ gives
\begin{equation} \label{eq:optimal W}
    \mathbf{W}_{k}^{\text{opt}}= \mathbf{E}_{k}^{-1}.
\end{equation}

Substituting $\mathbf{u}_{k}^{\text{opt}}\left( \mathbf{r} \right)$ and $\mathbf{W}_{k}^{\text{opt}}$ into~\eqref{P2:MMSE} yields an optimization problem with respect to $\mathbf{v}_k(\mathbf{s})$ as
\begin{subequations}\label{P3:WSR-CAPA}
\begin{align}
&\!\!\!\max_{\mathbf{v}_k(\mathbf{s})} 
\sum\limits_{k=1}^{K} \log\det\big(\big(\mathbf{E}_k^{\text{opt}}\big)^{\!-1}\big)\label{P3:objective} \\
&\text{s.t.}\quad\!\!\! \mathbf{E}_{k}^{\text{opt}}\!\!=\!\mathbf{I}_d\!-\!\!\!\iint_{\mathcal{S} _{\mathrm{U}}^k}\!\!\!\!\!\!{\mathbf{a}_{kk}^{\mathsf H}(\mathbf{r}_1)\mathrm{J}_{k}^{-1}\!\!\left( \mathbf{r}_1,\mathbf{r}_2 \right) \mathbf{a}_{kk}(\mathbf{r}_2)}\mathrm{d}\mathbf{r}_1\mathrm{d}\mathbf{r}_2,\label{P3:E opt}\\
&\quad\,\,\,\,\,\,\,\eqref{eq:theorem1 a},\eqref{P1:power constraint}.\nonumber
\end{align}
\end{subequations}

Applying Lemma~\ref{lem:woodbury} in Appendix~\ref{appendix_capacity} to the inverse kernel in~\eqref{P3:E opt} gives
\begin{equation}\label{eq:inverse of J}
\begin{aligned}
&\mathrm{J}_{k}^{-1}\left( \mathbf{r}_1,\mathbf{r}_2 \right) =\left(\mathrm{J}_{\bar{k}}\left( \mathbf{r}_1,\mathbf{r}_2 \right) +\mathbf{a}_{kk}(\mathbf{r}_1)\mathbf{a}_{kk}^{\mathsf H}\left( \mathbf{r}_2 \right) \right) ^{-1}
\\
&\quad=\mathrm{J}_{\bar{k}}^{-1}\left( \mathbf{r}_1,\mathbf{r}_2 \right) -\bm{\uppsi}_k(\mathbf{r}_1)\left( \mathbf{I}_d+\mathbf{G}_k \right)^{-1} \bm{\upphi}_k(\mathbf{r}_2),
\end{aligned}
\end{equation}
where $\mathbf{G}_k=\iint_{\mathcal{S} _{\mathrm{U}}^{k}}{\mathbf{a}_{kk}^{\mathsf H}\left( \mathbf{r}_1 \right) \mathrm{J}_{\bar{k}}^{-1}\left( \mathbf{r}_1,\mathbf{r}_2 \right)\mathbf{a}_{kk}(\mathbf{r}_2)}\mathrm{d}\mathbf{r}_1\mathrm{d}\mathbf{r}_2$, and $\mathrm{J}_{\bar{k}}(\mathbf{r}_1,\mathbf{r}_2)$, $\bm{\uppsi}_k\left( \mathbf{r}_1 \right)$ and $\bm{\upphi}_k(\mathbf{r}_2)$ are defined as
\begin{subequations}\label{eq:auxiliary terms}
\begin{align}
&   \mathrm{J}_{\bar{k}}(\mathbf{r}_1,\mathbf{r}_2)\!\!=\!\! \sum\limits_{j=1,j\neq k}^K\!\! \mathbf{a}_{kj}(\mathbf{r}_1)\mathbf{a}_{kj}^{\mathsf{H}}(\mathbf{r}_2) 
\!+\!{\sigma_n^{2}}\delta(\mathbf{r}_1\!-\!\mathbf{r}_2),\\
&\bm{\uppsi}_k\left( \mathbf{r}_1 \right) =    \int_{\mathcal{S} _{\mathrm{U}}^{k}}{\mathrm{J}_{\bar{k}}^{-1}\left( \mathbf{r}_1,\mathbf{r}_2 \right) \mathbf{a}_{kk}(\mathbf{r}_2)\mathrm{d}\mathbf{r}_2}, \\
&\bm{\upphi}_k(\mathbf{r}_2)=    \int_{\mathcal{S} _{\mathrm{U}}^{k}}{\mathbf{a}_{kk}^{\mathsf H}\left( \mathbf{r}_1 \right) \mathrm{J}_{\bar{k}}^{-1}\left( \mathbf{r}_1,\mathbf{r}_2 \right) \mathrm{d}\mathbf{r}_1}.
\end{align}
\end{subequations}
Substituting~\eqref{eq:inverse of J} and~\eqref{eq:auxiliary terms} into~\eqref{P3:E opt} yields
\begin{equation}\label{eq:inverse of E}
\begin{aligned}
&(\mathbf{E}_k^{\text{opt}})^{-1}=\big( \mathbf{I}_d-\mathbf{G}_k\left( \mathbf{I}_d+\mathbf{G}_k \right) ^{-1} \big) ^{-1}\!\!\!=\mathbf{I}_d+\mathbf{G}_k
\\
&   =\mathbf{I}_d+\!\! \iint_{\mathcal{S}^k_{\mathrm{U}}}{\!}\mathbf{a}_{kk}^{\mathsf{H}}(\mathbf{r}_1)\,\mathrm{J}_{\bar{k}}^{-1}(\mathbf{r}_1,\mathbf{r}_2)\,\mathbf{a}_{kk}(\mathbf{r}_2)\,\mathrm{d}\mathbf{r}_2\,\mathrm{d}\mathbf{r}_1.
\end{aligned}
\end{equation}
Combining~\eqref{eq:inverse of E} with~\eqref{P3:WSR-CAPA} completes the proof.

\section{Proof of Proposition~\ref{prop:wmmse_recursion}}
\label{appendix:wmmse_recursion}
To prove Proposition~\ref{prop:wmmse_recursion}, we first apply Lemma~\ref{lem:woodbury} to the inverse kernels in~\eqref{eq:functional-update-v} and~\eqref{eq:functional update u}, and then substitute the resulting expression of~\eqref{eq:functional update u} into~\eqref{eq:functional-update-v}.

Define $\mathbf{m}\left( \mathbf{s} \right) \triangleq\left[\mathbf{c}_1\!\left( \mathbf{s} \right),\cdots ,\mathbf{c}_K\!\left( \mathbf{s} \right) \right] \in \mathbb{C} ^{1\times Kd}$ and $\mathbf{n}\left( \mathbf{s} \right) \triangleq\left[\mathbf{c}_{1}\left( \mathbf{s} \right)\mathbf{W}_1^{\mathsf{H}} ,\cdots ,\mathbf{c}_{K}\left( \mathbf{s} \right)\mathbf{W}^{\mathsf{H}}_K \right]^{\mathsf{H}} \in \mathbb{C} ^{Kd\times 1}$. Then, 
\begin{equation} \label{eq:T eq mn}
    \mathrm{T}_k\left( \mathbf{s}_1,\mathbf{s} \right)
=
\mu \delta \left( \mathbf{s}_1-\mathbf{s} \right)
+\mathbf{m}\left( \mathbf{s}_1 \right)\mathbf{n}\left( \mathbf{s} \right).
\end{equation}
Substituting~\eqref{eq:T eq mn} into~\eqref{eq:functional-update-v} and applying Lemma~\ref{lem:woodbury} gives
\begin{equation}\label{eq:update equation of Vk_appendix}
\begin{split}\raisetag{0.6cm}
&\mathbf{v}_{k}\left( \mathbf{s} \right) = \!\int_{\mathcal{S} _{\mathrm{B}}}\!\!\Big(\frac{1}{\mu}\delta \left( \mathbf{s}_1-\mathbf{s} \right)-\\
&\quad\quad\quad\quad\frac{1}{\mu}\mathbf{m}\left( \mathbf{s} \right)
\left( \mathbf{I}_{Kd}+\mathbf{G} \right)^{-1}
\frac{1}{\mu}\mathbf{n}\left( \mathbf{s}_1 \right)\Big)\mathbf{c}_{k}\left( \mathbf{s}_1 \right)\mathbf{W}_k \mathrm{d}\mathbf{s}_1\\
&=
\frac{1}{\mu}\mathbf{c}_k\left( \mathbf{s} \right) \mathbf{W}_k\!\!-\!\!\frac{1}{\mu ^2}\mathbf{m}\left( \mathbf{s} \right)
\left( \mathbf{I}_{Kd}\!+\!\mathbf{G} \right) ^{-1}\!\!\!\!\int_{\mathcal{S} _{\mathrm{B}}}\!\!\!\!\!\mathbf{n}\left( \mathbf{s}_1 \right) \mathbf{c}_k\left( \mathbf{s}_1 \right) \mathbf{W}_k
\mathrm{d}\mathbf{s}_1\\
&\triangleq
\frac{1}{\mu}\mathbf{c}_k\left( \mathbf{s} \right) \mathbf{W}_k
-\mathbf{m}\left( \mathbf{s} \right){\mathbf{D}}_k,
\end{split}
\end{equation}
where ${\mathbf D}_k
\triangleq
\frac{1}{\mu^{2}}
\left(\mathbf I_{Kd}+\mathbf G\right)^{-1}
\int_{\mathcal S_{\mathrm B}}
\mathbf n(\mathbf s)\mathbf c_k(\mathbf s)\mathbf W_k
\mathrm d\mathbf s$ and $\mathbf{G}=\frac{1}{\mu}\int_{\mathcal{S}_{\mathrm{B}}}\mathbf{n}\left( \mathbf{s}\right)\mathbf{m}\left(\mathbf{s}\right)\mathrm{d}\mathbf{s}$. Since ${\mathbf D}_k\in\mathbb C^{Kd\times d}$, it can be partitioned into $K$ stacked block matrices as ${\mathbf D}_k=
\left[\bar{\mathbf D}_{k1}^{\mathsf T},\ldots,\bar{\mathbf D}_{kK}^{\mathsf T}\right]^{\mathsf T}$ with $\bar{\mathbf D}_{ki}\in\mathbb C^{d\times d},\ \forall i$. 
Using the definition of $\mathbf m(\mathbf s)$ in~\eqref{eq:T eq mn},~\eqref{eq:update equation of Vk_appendix} can be rewritten as
\begin{equation}\label{eq:update equation of Vk2_appendix}
\begin{split}
\mathbf v_k(\mathbf s)&=\frac{1}{\mu}\mathbf{c}_k\left( \mathbf{s} \right) \mathbf{W}_k
-\sum\limits_{i=1}^{K}\mathbf c_i(\mathbf s)\bar{\mathbf D}_{ki}\\
&=\sum\limits_{i=1}^{K}\mathbf c_i(\mathbf s)\mathbf F_{ki},
\end{split}
\end{equation}
where $\mathbf F_{kk}=\frac{1}{\mu}\mathbf W_k-\bar{\mathbf D}_{kk}$ and $\mathbf F_{ki}=-\bar{\mathbf D}_{ki}$ for $i\neq k$. Substituting the definition of $\mathbf c_i(\mathbf s)$ in~\eqref{eq:c function} into~\eqref{eq:update equation of Vk2_appendix} yields
\begin{equation}\label{eq:update equation of Vk3_appendix}
\mathbf v_k(\mathbf s)
=
\sum\limits_{i=1}^{K}
\Big(
\int_{\mathcal S_{\mathrm U}^{i}}
h_i^{\mathsf H}(\mathbf r,\mathbf s)\mathbf u_i(\mathbf r)\,\mathrm d\mathbf r
\Big){\mathbf F}_{ki}.
\end{equation}

To handle the inverse kernel $\mathrm{J}_{k}\!\left( \mathbf{r}_1,\mathbf{r}_2 \right)$ in~\eqref{eq:functional update u}, define $\mathbf{o}_k\left( \mathbf{r} \right)\triangleq
\left[ \mathbf{a}_{k1}\left( \mathbf{r} \right) ,\cdots ,\mathbf{a}_{kK}\left( \mathbf{r} \right) \right]
\in \mathbb{C} ^{1\times Kd}$. Then, 
\begin{equation} \label{eq:J eq oo}
    \mathrm{J}_{k}\!\left( \mathbf{r}_1,\mathbf{r}_2 \right)=
\sigma_n^{2}\delta \left( \mathbf{r}_1-\mathbf{r}_2 \right)+\mathbf{o}_k\left( \mathbf{r}_1 \right)\mathbf{o}^{\mathsf{H}}_k\left( \mathbf{r}_2 \right).
\end{equation}
Substituting~\eqref{eq:J eq oo} into~\eqref{eq:functional update u} and applying Lemma~\ref{lem:woodbury} in Appendix~\ref{appendix_capacity} yields
\begin{equation}\label{eq:update equation of uk_appendix}
\begin{split}\raisetag{0.6cm}
&\mathbf{u}_{k}\left( \mathbf{r} \right)
= \int_{\mathcal{S}_{\mathrm{U}}^{k}}
\Big(
\frac{1}{\sigma_n^2}\delta \left( \mathbf{r}_1-\mathbf{r} \right)
-\\
&\quad\quad\quad\quad\frac{1}{\sigma_n^4}\mathbf{o}_k\left( \mathbf{r} \right)
\left( \mathbf{I}_{Kd}+\bar{\mathbf{G}}_k \right)^{-1}
\mathbf{o}_k^{\mathsf{H}}\left( \mathbf{r}_1 \right)
\Big)
\mathbf{a}_{kk}\left( \mathbf{r}_1 \right)\mathrm{d}\mathbf{r}_1 \\
&=\frac{1}{\sigma_n^2}\mathbf{a}_{kk}\left( \mathbf{r} \right)\!\!-\!\!\frac{1}{\sigma_n^4}\mathbf{o}_k\!\left( \mathbf{r} \right)\!\left( \mathbf{I}_{Kd}\!+\!\bar{\mathbf{G}}_k \right)^{-1}\!\!\!\!
\int_{\mathcal{S}_{\mathrm{U}}^{k}}\!\!\!\!\!\mathbf{o}_{k}^{\mathsf H}\left( \mathbf{r}_1 \right)
\mathbf{a}_{kk}\left( \mathbf{r}_1 \right)
\mathrm{d}\mathbf{r}_1\\
&\triangleq\frac{1}{\sigma_n^2}\mathbf{a}_{kk}\left( \mathbf{r} \right)\!\!-\!\!\frac{1}{\sigma_n^4}\mathbf{o}_k\!\left( \mathbf{r} \right)\mathbf{Z}_{k},
\end{split}
\end{equation}
where $\mathbf{Z}_{k}\triangleq\left( \mathbf{I}_{Kd}\!+\!\bar{\mathbf{G}}_k \right)^{-1}\!\!
\int_{\mathcal{S}_{\mathrm{U}}^{k}}\!\mathbf{o}_{k}^{\mathsf H}\left( \mathbf{r}_1 \right)
\mathbf{a}_{kk}\left( \mathbf{r}_1 \right)
\mathrm{d}\mathbf{r}_1$ and $\bar{\mathbf{G}}_k=\frac{1}{\sigma_n^2}\int_{\mathcal{S}^k_{\mathrm U}}\mathbf{o}_{k}^{\mathsf H}\left( \mathbf{r} \right)
\mathbf{o}_k\left( \mathbf{r} \right)
\mathrm{d}\mathbf{r}$. Since ${\mathbf Z}_k\in\mathbb C^{Kd\times d}$, it can be partitioned into $K$ stacked block matrices as ${\mathbf Z}_k=
\left[\bar{\mathbf Z}_{k1}^{\mathsf T},\ldots,\bar{\mathbf Z}_{kK}^{\mathsf T}\right]^{\mathsf T}$ with $\bar{\mathbf Z}_{ki}\in\mathbb C^{d\times d},\ \forall i$. 
Using the definition of $\mathbf{o}_k(\mathbf r)$ above~\eqref{eq:J eq oo},~\eqref{eq:update equation of uk_appendix} can be rewritten as
\begin{equation}\label{eq:update equation of uk2_appendix}
\begin{split}
\mathbf u_k(\mathbf r)&=\frac{1}{\sigma_n^2}\mathbf{a}_{kk}\left( \mathbf{r} \right)
-\sum\limits_{i=1}^{K}\mathbf a_{ki}(\mathbf r)\bar{\mathbf Z}_{ki}\\
&=\sum\limits_{i=1}^{K}\mathbf a_{ki}(\mathbf r)\mathbf L_{ki},
\end{split}
\end{equation}
where $\mathbf L_{kk}=\frac{1}{\sigma_n^2}\mathbf{I}_{d}-\bar{\mathbf Z}_{kk}$ and $\mathbf L_{ki}=-\bar{\mathbf Z}_{ki}$ for $i\neq k$.

Using~\eqref{eq:update equation of uk2_appendix} and~\eqref{eq:update equation of Vk3_appendix}, the updates of $\mathbf{u}^{(l+1)}_k(\mathbf r)$ and $\mathbf v^{(l+1)}_k(\mathbf s)$ at iteration $(l+1)$ in Table~\ref{tab:functional-wmmse} can be written as
\begin{subequations}
    \begin{align}
\mathbf {u}^{(l+1)}_k(\mathbf r)&=\sum\limits_{i=1}^{K}\mathbf a^{(l)}_{ki}(\mathbf r)\mathbf L^{(l)}_{ki},\label{eq:u update equation at l+1 iteration}\\
\mathbf v^{(l+1)}_k(\mathbf s)&
=
\sum\limits_{i=1}^{K}
\Big(\!
\int_{\mathcal S_{\mathrm U}^{i}}
\!\!\!\!\!h_i^{\mathsf H}(\mathbf r,\mathbf s)\mathbf u^{(l+1)}_i\!(\mathbf r)\,\mathrm d\mathbf r
\Big)\mathbf F^{(l)}_{ki},\label{eq:v update equation at l+1 iteration}
    \end{align}
\end{subequations}
where $\mathbf{a}^{(l)}_{ij}\!\left( \mathbf{r} \right)\!\!=\!\!\int_{\mathcal{S}_{\mathrm{B}}}
\!\!h_i\left( \mathbf{r},\mathbf{s}_1 \right)\mathbf{v}_{j}^{(l)}\!\!\left( \mathbf{s}_1 \right)\mathrm{d}\mathbf{s}_1$, and $\mathbf L^{(l)}_{ki}$ and $\mathbf F^{(l)}_{ki}$ are the corresponding coefficient matrices at iteration $(l+1)$.

Substituting~\eqref{eq:u update equation at l+1 iteration} into~\eqref{eq:v update equation at l+1 iteration} gives
\begin{equation}\label{eq:iterative equation of vk_appendix}
\begin{split}\raisetag{1cm}
&\mathbf{v}_{k}^{(l+1)}(\mathbf{s})
= \sum_{i=1}^{K}\!\!\int_{\mathcal{S}_{\mathrm{U}}^{i}}\!\!\!\!\!h_{i}^{\mathsf{H}}(\mathbf{r},\mathbf{s})\big(\sum\limits_{j=1}^{K}\mathbf a^{(l)}_{ij}(\mathbf r)\mathbf L^{(l)}_{ij}\big) \,\mathrm d\mathbf{r}\;\mathbf{F}^{(l)}_{ki}\\
&=\int_{\mathcal{S}_{\mathrm{U}}^{k}}h_{k}^{\mathsf{H}}\left( \mathbf{r},\mathbf{s} \right) \mathbf{a}^{(l)}_{kk}\left( \mathbf{r}\right)\mathrm{d}\mathbf{r}\mathbf{\Theta}_{k}\\
&\quad+\sum_{i=1}^K\sum_{\substack{j=1\\ (i,j)\ne (k,k)}}^K
\int_{\mathcal{S} _{\mathrm{U}}^{i}}h_{i}^{\mathsf{H}}\left( \mathbf{r},\mathbf{s} \right) \mathbf{a}^{(l)}_{ij}\left( \mathbf{r} \right)\mathrm{d}\mathbf{r}\mathbf{\Sigma}^k_{ij},
\end{split}
\end{equation}
where $\mathbf{\Theta}_k=\mathbf{L}^{(l)}_{kk}\mathbf{F}^{(l)}_{kk}$ and $\mathbf{\Sigma}^k_{ij}=\mathbf{L}^{(l)}_{ij}\mathbf{F}^{(l)}_{ki}$.

\bibliographystyle{IEEEtran}
\bibliography{main}
\end{document}